\documentclass[11pt,letterpaper]{article}

\usepackage{amsmath,amsfonts,amssymb,color}
\usepackage[dvipsnames]{xcolor}
\definecolor{winered}{rgb}{0.8,0,0}
\definecolor{myblue}{rgb}{0,0,0.8}
\usepackage{tikz}
\usepackage{pmat}
\usepackage{amsthm}
\usepackage{empheq}

\newtheorem{definition}{Definition}
\newtheorem{theorem}{Theorem}
\newtheorem{lemma}{Lemma}

\newtheorem{corollary}{Corollary}
\newtheorem{remark}{Remark}
\newtheorem{assumption}{Assumption}

\usepackage{epsfig} 
\DeclareMathOperator*{\argmax}{\arg\!\max}
\DeclareMathOperator*{\argmin}{\arg\!\min}

\usepackage{algorithm}
\usepackage{algpseudocode}

\usepackage{epsfig}
\usepackage{array}
\usepackage{multirow}
\usepackage{epstopdf}
\usepackage{tikz}
\usepackage{relsize}
\usetikzlibrary{shapes,arrows}
\usepackage[hidelinks]{hyperref}
\hypersetup{
    colorlinks=true,
    linkcolor={winered},
    citecolor={myblue}
}
\usepackage{cite}
\usepackage[margin=1in]{geometry} 
\begin{document}
%
\title{A New Approach to Distributed Hypothesis Testing and Non-Bayesian Learning: Improved Learning Rate and Byzantine-Resilience} 
%
%
%

\author{Aritra Mitra, John A. Richards and Shreyas Sundaram
\thanks{A. Mitra and S. Sundaram are with the School of Electrical and Computer Engineering at Purdue University. J. A.  Richards is with Sandia National Laboratories.   Email: {\tt \{mitra14, sundara2\}@purdue.edu},  {\tt{jaricha@sandia.gov}}. This work was supported in part by NSF CAREER award
1653648, and by the Laboratory Directed Research and Development program at Sandia National Laboratories. Sandia National Laboratories is a multimission laboratory managed and operated by National Technology \& Engineering Solutions of Sandia, LLC, a wholly owned subsidiary of Honeywell International Inc., for the U.S. Department of Energy's National Nuclear Security Administration under contract DE-NA0003525. The views expressed in the article do not necessarily represent the views of the U.S. Department of Energy or the United States Government.}}

\maketitle
\begin{abstract}
We study a setting where a group of agents, each receiving partially informative private signals, seek to collaboratively learn the true
underlying state of the world (from a finite set of hypotheses) that generates their joint observation profiles.  To solve this problem, we propose a distributed learning rule that differs fundamentally from existing approaches, in that it does not employ any form of ``belief-averaging". Instead, agents update their beliefs based on a min-rule. Under standard assumptions on the observation model and the network structure, we establish that each agent learns the truth asymptotically almost surely. As our main contribution, we prove that with probability 1, each false hypothesis is ruled out by every agent exponentially fast at a network-independent rate that is strictly larger than existing rates. We then develop a computationally-efficient variant of our learning rule that is provably resilient to agents who do not behave as expected (as represented by a Byzantine adversary model) and deliberately try to spread misinformation.  
\end{abstract}
\maketitle
\section{Introduction}
Given noisy data, the task of making meaningful inferences about a quantity of interest is at the heart of various complex estimation and detection problems arising in signal processing, information theory, machine learning, and control systems. When the information required to solve such problems is dispersed over a network, several interesting questions arise. How should the individual entities in the network combine their own private observations with the information received from neighbors to learn the quantity of interest? What are the minimal requirements on the information structure of the entities and the topology of the network for this to happen? How fast does information spread as a function of the diffusion rule and the structure of the network? What can be said when the underlying network changes with time and/or certain entities deviate from nominal behavior? In this paper, we provide rigorous theoretical answers to such questions for the setting where a group of agents receive a stream of private signals generated by an unknown quantity known as the ``true state of the world". Communication among such agents is modeled by a graph. The goal of each agent is to eventually identify the true state from a finite set of hypotheses. However, while the \textit{collective} signals across all agents might facilitate identification of the true state, signals received by any given agent may, in general, not be rich enough for identifying the state in isolation.  Thus, the problem of interest is to develop and analyze local interaction rules that facilitate inference of the true state at every agent. The setup described above serves as a common mathematical abstraction for modeling and analyzing various decision-making problems in social and economic networks (e.g., opinion formation and spreading), and classification/detection problems arising in large-scale engineered systems (e.g., object recognition by a group of aerial robots). While the former is typically studied under the moniker of non-Bayesian social learning, the latter usually goes by the name of distributed detection/hypothesis testing. In what follows, we discuss relevant literature.

\textbf{Related Literature}: Much of the earlier work on this topic of interest assumed the existence of a centralized fusion center for performing computational tasks \cite{fusion1,fusion2,fusion3}. Our work in this paper, however, belongs to a more recent body of literature wherein  individual agents are endowed with computational capabilities, and interactions among them are captured by a graph \cite{GEBjad,jad2,liu,rad,shahinparam,shahinTAC,nedic,nedic2,lalitha1,lalitha2,su1}. These works are essentially inspired by the model in \cite{GEBjad}, where each agent maintains a belief vector (over the set of hypotheses) that is sequentially updated as the convex combination of its own Bayesian posterior and the priors of its neighbors. Subsequent approaches share a common theme: they typically involve a learning rule that combines a local Bayesian update with a consensus-based opinion pooling of neighboring beliefs. The key point of distinction among such rules stems from the specific manner in which neighboring opinions are aggregated. Specifically, linear opinion pooling is studied in \cite{GEBjad,jad2,liu}, whereas log-linear opinion pooling is studied in \cite{rad,shahinparam,shahinTAC,nedic,nedic2,lalitha1,lalitha2,su1}. Under appropriate conditions on the observation model and the network structure, each of these approaches enable every agent to learn the true state exponentially fast, with probability 1. The rate of convergence, however, depends on the specific nature of the learning rule. Notably, finite-time concentration results are derived in \cite{nedic,nedic2,shahinTAC}, and a large-deviation analysis is conducted in \cite{lalitha1,lalitha2} for a broad class of distributions that generate the agents' observation profiles. Extensions to different types of time-varying graphs have also been considered in \cite{liu,nedic,nedic2,shahinparam,shahinTAC}. In a recent paper \cite{molavi}, the authors go beyond specific functional forms of belief-update rules and, instead, adopt an axiomatic framework that identifies the fundamental factors responsible for social learning. We point out that belief-consensus algorithms on graphs have been studied prior to \cite{GEBjad} as well as in \cite{olfati,saligrama}. The model in \cite{olfati,saligrama} differs from that in \cite{GEBjad,jad2,liu,nedic,nedic2,lalitha1,lalitha2,shahinparam,shahinTAC,rad,su1} in one key aspect: while in the former each agent has access to only one observation, the latter allows for influx of new information into the network in the form of a time-series of observations at every agent. 

\textbf{Our Contributions}: In light of the above developments, we now elaborate on the main contributions of this work. 

\textbf{1) A Novel Distributed Learning Rule}: In \cite[Section III]{nedic}, the authors explain that the commonly studied linear and log-linear forms of belief aggregation are specific instances of a more general class of opinion pooling known as g-Quasi-Linear Opinion pools (g-QLOP), introduced in \cite{pools}. Our first contribution is the development of a novel belief update rule that deviates fundamentally from the broad family of g-QLOP learning rules. Specifically, the learning algorithm that we propose in Section \ref{sec:Algo} does not rely on any linear consensus-based belief aggregation protocol. Instead, each agent maintains two sets of belief vectors: a local belief vector and an actual belief vector. Each agent updates its local belief vector in a  Bayesian manner based on only its private observations, i.e., without the influence of neighbors. The actual belief on each hypothesis is updated (up to normalization) as the \textit{minimum} of the agent's own local belief and the actual beliefs of its neighbors on that particular hypothesis. We provide theoretical guarantees on the performance of this algorithm in Section \ref{sec:Results}. As we explain later in the paper, establishing such guarantees requires proof techniques that differ substantially from those existing. 

\textbf{2) Strict Improvement in Rate of Learning}: While data-aggregation via arithmetic or geometric averaging of neighboring beliefs allows asymptotic learning, such schemes may potentially dilute the rate at which false hypotheses are eliminated. In particular, for the linear consensus protocol introduced in  \cite{GEBjad}, the limiting rate at which a particular false hypothesis is eliminated is almost surely upper-bounded by a quantity that depends on the relative entropies and centralities of the agents \cite{jad2}. The log-linear rules in \cite{nedic,nedic2,lalitha1,lalitha2,shahinTAC} improve upon such a rate: with probability 1, the asymptotic rate of rejection of a false hypothesis under such rules is a convex combination of the agents' relative entropies, where the convex weights correspond to the eigenvector centralities of the agents. In contrast, based on our approach, each false hypothesis is rejected by every agent exponentially fast, at a rate that is almost surely lower-bounded by the \textit{best} relative entropy (between the true state and the false hypothesis) among all agents, provided the underlying network is static and strongly-connected. In Theorem \ref{thm:main}, we show that the above result continues to hold even when the network changes with time, as long as a mild joint strong-connectivity condition is met. \textit{Thus, to the best of our knowledge, our approach leads to a \textit{strict improvement} in the rate of learning over all existing approaches: this constitutes our main contribution}. 

\textbf{3) Resilience to Adversaries}: Despite the wealth of literature on distributed inference, there is limited understanding of the impact of misbehaving agents who do not follow the prescribed learning algorithm. Such agents may represent stubborn individuals or ideological extremists in the context of a social network, or model faults (either benign or malicious) in a networked control system. \textit{In the presence of such misbehaving entities, how should the remaining agents process their private observations and the beliefs of their neighbors to eventually learn the truth?}  To answer this question, we capture deviant behavior via the classical Byzantine adversary model \cite{lynch}, and develop a provably correct, resilient version of our proposed learning rule in Section \ref{sec:LFRHE}. Theorem \ref{thm:Byz} characterizes the performance of this rule and, in particular, reveals that each regular agent can infer the truth exponentially fast. Furthermore, we identify conditions on the observation model and the network structure that guarantee applicability of our Byzantine-resilient learning rule, and argue that such conditions can be checked in polynomial time. The only related work that we are aware of in this regard is \cite{su1}. As we discuss in detail in Section \ref{sec:LFRHE}, our proposed approach has various computational advantages relative to those in \cite{su1}. 

In addition to the main contributions discussed above, a minor contribution of this paper is the following. For static graphs where all agents behave normally, Theorem \ref{thm:minimal} establishes consistency of our learning rule under conditions that are necessary for \textit{any} belief update rule to work, when agents make conditionally independent observations. In particular, we show that the typical assumption of strong-connectivity on the network can be relaxed, and identify the minimal requirement for uniquely learning any state that gets realized.\footnote{A strongly-connected graph has a path between every pair of nodes.}  Despite its various advantages, our approach cannot, in general, handle the scenario where there does not exist any single true state that generates signals consistent with those seen by every agent. The method in \cite{nedic,nedic2}, however, is applicable to this case as well, and enables each agent to identify the hypothesis that \textit{best} explains the groups' observations. 

A preliminary version of this paper appeared as \cite{mitraACC19}. We significantly expand upon the content in \cite{mitraACC19}
by (i) providing detailed convergence rate analyses of our algorithms, (ii) extending our results to the case of time-varying graphs, (iii) elaborating on the significance of our results relative to prior work, and validating them via suitable simulation studies. 

\section{Model and Problem Formulation}
\label{sec:model}
\textbf{Network Model:} Let $\mathbb{N}$ and $\mathbb{N}_{+}$ denote the set of non-negative integers and positive integers, respectively. We consider a group of  agents $\mathcal{V}=\{1,2,\ldots,n\}$ interacting over a time-varying, directed communication graph $\mathcal{G}[t]=(\mathcal{V},\mathcal{E}[t])$, where $t\in\mathbb{N}$. An edge $(i,j)\in\mathcal{E}[t]$ indicates that agent $i$ can directly transmit information to agent $j$ at time-step $t$. If $(i,j)\in\mathcal{E}[t]$, then at time $t$,  agent $i$ will be called a neighbor of agent $j$, and agent $j$ will be called an out-neighbor of agent $i$. The set $\mathcal{N}_i[t]$ will be used to denote the neighbors of agent $i$ (excluding itself) at time $t$, whereas the set $\mathcal{N}_i[t]\cup\{i\}$ will be referred to as the inclusive neighborhood of agent $i$ at time $t$. We will use $|\mathcal{C}|$ to denote the cardinality of a set $\mathcal{C}$.

\textbf{Observation Model:} Let $\Theta=\{\theta_1,\theta_2,\ldots,\theta_m\}$ denote $m$ possible states of the world; each $\theta_i\in\Theta$ will be called a hypothesis. At each time-step $t\in\mathbb{N}_{+}$, every agent $i\in\mathcal{V}$  privately observes a signal $s_{i,t}\in\mathcal{S}_i$, where $\mathcal{S}_i$ denotes the signal space of agent $i$. The joint observation profile so generated across the network is denoted ${s}_{t}=(s_{1,t},s_{2,t},\ldots,s_{n,t})$, where $s_t\in\mathcal{S}$, and $\mathcal{S}=\mathcal{S}_1\times\mathcal{S}_2\times\ldots \mathcal{S}_n$. 
The signal $s_{t}$ is generated based on a conditional likelihood function $l(\cdot|\theta^{\star})$, governed by the true state of the world $\theta^{\star}\in\Theta$. Let $l_i(\cdot|\theta^{\star}), i\in\mathcal{V}$ denote the $i$-th marginal of $l(\cdot|\theta^{\star})$. 
The signal structure of each agent $i\in\mathcal{V}$ is then characterized by a family of parameterized marginals $\{l_i(w_i|\theta): \theta\in\Theta, w_i\in\mathcal{S}_i\}$.\footnote{Whereas $w_i\in\mathcal{S}_i$ will be used to refer to a generic element of the signal space of agent $i$,  $s_{i,t}$  will denote the random variable (with distribution $l_i(\cdot|\theta^{\star})$) that corresponds to the observation of agent $i$ at time-step $t$.} 

We make the following standard assumptions \cite{GEBjad,jad2,liu,nedic,nedic2,lalitha1,lalitha2,shahinparam,shahinTAC,su1}: (i) The signal space of each agent $i$, namely $\mathcal{S}_i$, is finite.\footnote{The analysis in \cite{rad} applies to continuous parameter spaces.} (ii) Each agent $i$ has knowledge of its local likelihood functions $\{l_i(\cdot|\theta_p)\}_{p=1}^{m}$, and it holds that $l_i(w_i|\theta) > 0, \forall w_i\in\mathcal{S}_i$, and $\forall \theta \in \Theta$. (iii) The observation sequence of each agent is described by an i.i.d. random process over time; however, at any given time-step, the observations of different agents may potentially be correlated. (iv) There exists a fixed true state of the world $\theta^{\star}\in\Theta$ (unknown to the agents) that generates the observations of all the agents.\footnote{The approach in \cite{nedic} applies to a more general setting where there may not exist such a true hypothesis.} Finally, we define a probability triple $(\Omega,\mathcal{F},\mathbb{P}^{\theta^{\star}})$, where $\Omega\triangleq\{\omega: \omega=(s_1,s_2,\ldots),  s_t\in\mathcal{S}, t \in \mathbb{N}_{+}\}$, $\mathcal{F}$ is the $\sigma$-algebra generated by the observation profiles, and $\mathbb{P}^{\theta^{\star}}$ is the probability measure induced by sample paths in $\Omega$. Specifically, $\mathbb{P}^{\theta^{\star}}=\prod \limits_{t=1}^{\infty}l(\cdot|\theta^{\star})$. For the sake of brevity, we will say that an event occurs almost surely to mean that it occurs almost surely w.r.t. the probability measure $\mathbb{P}^{\theta^{\star}}$. 

Note that assumptions (i) and (ii) on the observation model imply the existence of a constant $L\in(0,\infty)$ such that:
\begin{equation}
    \max_{i\in\mathcal{V}} \max_{w_i\in\mathcal{S}_i} \max_{\theta_p,\theta_q\in\Theta} \left|\log\frac{l_i(w_i|\theta_p)}{l_i(w_i|\theta_q)}\right| \leq L.
\label{eqn:bounded_lograt}
\end{equation}
We will make use of the above fact later in our analysis.

Given the above setup, the goal of each agent in the network is to discern the true state of the world $\theta^{\star}$. The challenge associated with such a task stems from the fact that the private signal structure of any given agent is in general only partially informative. To make this notion precise, define $\Theta^{\theta^{\star}}_i\triangleq\{\theta\in\Theta : l_i(w_i|\theta)=l_i(w_i|\theta^{\star}), \forall w_i\in\mathcal{S}_i\}.$ In words, $\Theta^{\theta^{\star}}_i$ represents the set of hypotheses that are \textit{observationally equivalent} to the true state $\theta^{\star}$ from the perspective of agent $i$. In general, for any agent $i\in\mathcal{V}$, we may have $|\Theta^{\theta^{\star}}_i| > 1$, necessitating collaboration among agents subject to the restrictions imposed by the time-varying communication topology. 

Our \textbf{objective} in this paper will be to design a distributed learning rule that allows each agent $i\in\mathcal{V}$ to identify the true state of the world asymptotically almost surely. To this end, we now introduce the following notion of source agents that will be useful in our subsequent developments.

\begin{definition}(\textbf{Source agents}) An agent $i$ is said to be a source agent for a pair of distinct hypotheses $\theta_p,\theta_q\in\Theta$, if $D(l_i(\cdot|\theta_p)||l_i(\cdot|\theta_q)) > 0$, where $D(l_i(\cdot|\theta_p)||l_i(\cdot|\theta_q))$ represents the KL-divergence between the distributions $l_i(\cdot|\theta_p)$ and $l_i(\cdot|\theta_q)$, and is given by:
\begin{equation}
D(l_i(\cdot|\theta_p)||l_i(\cdot|\theta_q))=\sum \limits_{w_i\in\mathcal{S}_i}l_i(w_i|\theta_p)\log\frac{l_i(w_i|\theta_p)}{l_i(w_i|\theta_q)}.
\end{equation}
The set of all source agents for the pair $\theta_p,\theta_q$ is denoted by $\mathcal{S}(\theta_p,\theta_q)$.
\end{definition}

In words, a source agent for a pair $\theta_p,\theta_q\in\Theta$ is an agent that can distinguish between the pair of hypotheses $\theta_p,\theta_q$ based on its private signal structure. It should be noted that $\mathcal{S}(\theta_p,\theta_q)=\mathcal{S}(\theta_q,\theta_p)$, since $D(l_i(\cdot|\theta_p)||l_i(\cdot|\theta_q))>0 \iff D(l_i(\cdot|\theta_q)||l_i(\cdot|\theta_p))>0$ \cite{cover}. To avoid cluttering the exposition, we will henceforth use $K_i(\theta_p,\theta_q)$ as a shorthand for $D(l_i(\cdot|\theta_p)||l_i(\cdot|\theta_q))$. In this work, we will assume that each state $\theta\in\Theta$ is \textit{globally identifiable} w.r.t. the joint observation model of the entire network. Based on our terminology of source agents, this translates to the following.

\begin{assumption}
(\textbf{Global Identifiability}) For each pair $\theta_p,\theta_q\in\Theta$ such that $\theta_p \neq \theta_q$, the set $\mathcal{S}(\theta_p,\theta_q)$ of agents that can distinguish between the pair $\theta_p,\theta_q$ is non-empty.
\label{ass:identifiability}
\end{assumption}

The above assumption is standard in the related literature. We will additionally make a mild assumption on the time-varying communication topology. To this end, let the union graph over an interval $[t_1,t_2], 0 \leq t_1 < t_2$, indicate a graph with vertex set equal to $\mathcal{V}$, and  edge set equal to $\bigcup^{t_2}_{\tau=t_1}\mathcal{E}[\tau]$. Based on this convention, we will assume (unless stated otherwise) that the sequence of communication graphs $\{\mathcal{G}[t]\}_{t=0}^{\infty}$ is \textit{jointly strongly-connected}, in the following sense.
\begin{assumption}\textbf{(Joint Strong-Connectivity)} There exists ${T}\in\mathbb{N}_{+}$ such that the union graph over every interval of the form $[rT,(r+1)T)$ is strongly-connected, where $r\in\mathbb{N}$.
\label{assump:connectivity}
\end{assumption}

While the above assumption on the network connectivity pattern is not necessary for solving the problem at hand, it is fairly standard in the analysis of distributed algorithms over time-varying networks \cite{jadTAC2013,nedicopt,nedic}. Having established the model and the problem formulation, we now proceed to a formal description of our distributed learning algorithm.
\section{Proposed Learning Rule}
\label{sec:Algo}
In this section, we propose a novel belief update rule (Algorithm \ref{Algo:Algo1}) and discuss the intuition behind it. Every agent $i$ maintains and updates (at every time-step $t$) two separate sets of belief vectors, namely, $\boldsymbol{\pi}_{i,t}$ and $\boldsymbol{\mu}_{i,t}$. Each of these vectors are probability distributions over the hypothesis set $\Theta$. We will refer to $\boldsymbol{\pi}_{i,t}$ and $\boldsymbol{\mu}_{i,t}$ as the ``local" belief vector (for reasons that will soon become obvious), and the ``actual" belief vector, respectively, maintained by agent $i$. The \textbf{goal} of each agent $i\in\mathcal{V}$ in the network will be to use its own private signals and the information available from its neighbors to update $\boldsymbol{\mu}_{i,t}$ sequentially, so that $\lim_{t\to\infty}\mu_{i,t}(\theta^{*})=1$ almost surely. To do so, at each time-step $t+1$ (where $t\in\mathbb{N}$), agent $i$ does the following for each $\theta\in\Theta$. It first generates $\pi_{i,t+1}(\theta)$ via a local Bayesian update rule that incorporates the private observation $s_{i,t+1}$ using $\pi_{i,t}(\theta)$ as a prior (line 5 in Algo. \ref{Algo:Algo1}). Having generated $\pi_{i,t+1}(\theta)$, agent $i$ updates $\mu_{i,t+1}(\theta)$ (up to normalization) by setting it to be the \emph{minimum} of its locally generated belief $\pi_{i,t+1}(\theta)$, and the actual beliefs $\mu_{j,t}(\theta), j\in \mathcal{N}_i[t]\cup\{i\}$ of its inclusive  neighborhood at the previous time-step (line 6 in Algo. \ref{Algo:Algo1}). It then reports $\boldsymbol{\mu}_{i,t+1}$ to each of its out-neighbors at time $t+1$.\footnote{Note that based on our algorithm, agents only exchange their actual beliefs, and not their local beliefs.} 

\begin{algorithm}[t]
	\caption{Belief update rule for each $i\in\mathcal{V}$}
\label{Algo:Algo1}
	\begin{algorithmic}[1]
	\State \textbf{Initialization:} $\mu_{i,0}(\theta)>0$,  $\pi_{i,0}(\theta)>0$, $\forall \theta\in\Theta$, and  $\sum_{\theta\in\Theta}\mu_{i,0}(\theta)=1$, $\sum_{\theta\in\Theta}\pi_{i,0}(\theta)=1$
	\State Transmit  $\boldsymbol{\mu}_{i,0}$ to  out-neighbors at time $0$
\For{$t+1\in\mathbb{N}_{+}$}
\For{$\theta\in\Theta$}
 \State Update local belief on $\theta$ as
   \begin{equation}
\pi_{i,t+1}(\theta)=\frac{l_i(s_{i,t+1}|\theta)\pi_{i,t}(\theta)}{\sum  \limits_{p=1}^{m} l_i(s_{i,t+1}|\theta_p)\pi_{i,t}(\theta_p)}
\label{eqn:Bayes}
\end{equation}

\State Update actual belief on $\theta$ as
\begin{equation}
\mu_{i,t+1}(\theta)=\frac{\min\{\{\mu_{j,t}(\theta)\}_{{j\in\mathcal{N}_i[t]\cup\{i\}}},\pi_{i,t+1}(\theta)\}}{\sum\limits_{p=1}^{m}\min\{\{\mu_{j,t}(\theta_p)\}_{{j\in\mathcal{N}_i[t]\cup\{i\}}},\pi_{i,t+1}(\theta_p)\}}
\label{eqn:rule1}
\end{equation}
\EndFor
 \State Transmit $\boldsymbol{\mu}_{i,t+1}$ to out-neighbors at time $t+1$
\EndFor
	\end{algorithmic}
\end{algorithm}

{\bf Intuition behind the learning rule}: 
At the core of our learning algorithm are two key principles: (1) \textit{Preservation of the intrinsic discriminatory capabilities of the agents}, and (2) \textit{Propagation of low beliefs on each false hypothesis}. We now elaborate on these features.

Consider the set of source agents $\mathcal{S}(\theta^{*},\theta)$ that can differentiate between a certain false hypothesis $\theta$ and the true state $\theta^{\star}$. By definition, the signal structures of such agents are rich enough for them to be able to eliminate $\theta$ on their own, i.e., without the support of their neighbors. To achieve this, we require each agent to maintain a local belief vector that is updated (via \eqref{eqn:Bayes}) \textit{without any network influence} using only the agent's own private signals. Doing so ensures that $\pi_{i,t}(\theta)\rightarrow 0$ a.s. for each $i\in\mathcal{S}(\theta^{\star},\theta)$. Next, leveraging this property, we want to be able to propagate low beliefs on $\theta$ from $\mathcal{S}(\theta^{\star},\theta)$ to $\mathcal{V}\setminus \mathcal{S}(\theta^{\star},\theta)$, i.e., the agents in $\mathcal{S}(\theta^{*},\theta)$ should contribute towards driving the actual beliefs of their out-neighbors (and eventually, of all the agents in the set $\mathcal{V}\setminus \mathcal{S}(\theta^{\star},\theta)$) on the hypothesis $\theta$ to zero. Using a min-rule of the form \eqref{eqn:rule1}, with $\pi_{i,t+1}(\theta)$ featuring as an external network-independent input, facilitates such propagation without compromising the abilities of agents in $\mathcal{S}(\theta^{\star},\theta)$ to eliminate $\theta$. When set in motion, our learning rule triggers a process of belief reduction on $\theta$  originating at $\mathcal{S(\theta^{\star},\theta)}$ that eventually propagates to each agent in the network reachable from $\mathcal{S(\theta^{\star},\theta)}$. 

\begin{remark}
We emphasize that the proposed learning rule given by Algorithm \ref{Algo:Algo1} does not employ any form of ``belief-averaging''. This feature is in stark contrast with existing approaches to distributed hypothesis testing that rely either on linear opinion pooling \cite{GEBjad,jad2,liu}, or log-linear opinion pooling\cite{rad,shahinparam,shahinTAC,nedic,nedic2,lalitha1,lalitha2,su1}. As such, the lack of linearity in our belief update rule precludes (direct or indirect) adaptation of existing analysis techniques to suit our needs. 
\end{remark}

\section{Analysis of Algorithm \ref{Algo:Algo1}}
\label{sec:Results}
\subsection{Statement of the Results}
In this section, we characterize the performance of Algorithm \ref{Algo:Algo1}. We start with one of the main results of the paper, proven in Appendix \ref{sec:App1}.
\begin{theorem}
Suppose the observation model satisfies the global identifiability condition (Assumption \ref{ass:identifiability}), and the sequence of communication graphs $\{\mathcal{G}[t]\}_{t=0}^{\infty}$ is jointly strongly-connected (Assumption \ref{assump:connectivity}).
Then, Algorithm \ref{Algo:Algo1} provides the following guarantees.
\begin{itemize}
    \item \textbf{(Consistency)}: For each agent  $i\in\mathcal{V}$,  $\mu_{i,t}(\theta^{\star}) \rightarrow 1$ a.s.
    \item \textbf{(Asymptotic Rate of Rejection of False Hypotheses)}: Consider any false hypothesis $\theta\in\Theta\setminus\{\theta^{\star}\}$. Then, the following holds for each agent $i\in\mathcal{V}$:
    \begin{equation}
        \liminf_{t\to\infty}-\frac{\log\mu_{i,t}(\theta)}{t} \geq \max_{v\in\mathcal{S}(\theta^{\star},\theta)}K_v(\theta^{\star},\theta) \hspace{1mm} a.s.
        \label{eqn:asymprate}
    \end{equation}
\end{itemize}
\label{thm:main}
\end{theorem}

The above result tells us that with probability $1$, every agent $i$ will be able to rule out each false hypothesis $\theta$ exponentially fast, at a rate that is eventually lower-bounded by the \textit{best} KL-divergence across the network between the pair of hypotheses $\theta^{\star}$ and $\theta$. In particular, this implies that given any $\epsilon >0$, the probability that agent $i$'s instantaneous rate of rejection of $\theta$, namely  $-\log\mu_{i,t}(\theta)/{t}$, is lower than the quantity $\max_{v\in\mathcal{S}(\theta^{\star},\theta)}K_v(\theta^{\star},\theta)$ by an additive factor of $\epsilon$, decays to zero. The next result, proven in Appendix \ref{sec:App2}, sheds some light on the rate of decay of this probability. 

\begin{theorem}
Suppose the conditions in Theorem \ref{thm:main} hold. Fix $\theta\in\Theta\setminus\{\theta^{\star}\}$, and let $\bar{K}(\theta^{\star},\theta)=\max_{v\in\mathcal{S}(\theta^{\star},\theta)}K_v(\theta^{\star},\theta)$. Then for every $\epsilon>0$ and $\delta\in(0,1)$, there exists a set $\Omega'(\delta)\subseteq\Omega$ with $\mathbb{P}^{\theta^{\star}}(\Omega'(\delta))\geq 1-\delta$, such that the following holds for each agent $i\in\mathcal{V}$:
\begin{equation}
   \liminf \limits_{t\to\infty} -\frac{1}{t}  \log \mathbb{P}^{\theta^{\star}}\left(\left\{-\frac{\log\mu_{i,t}(\theta)}{t} \leq \bar{K}(\theta^{\star},\theta)-\epsilon\right\}\cap\Omega'(\delta)\right) \geq \frac{\epsilon^2}{8L^2}.
\label{eqn:conc_main}
\end{equation}
\label{thm:conc}
\end{theorem}

Our next result pertains to the special case when the communication graph does not change over time, i.e., when $\mathcal{G}[t]=\mathcal{G}, \forall t\in\mathbb{N}$. To state the result, we will employ the following terminology. Given two disjoint sets $\mathcal{C}_1,\mathcal{C}_2 \subseteq{\mathcal{V}}$, we say $\mathcal{C}_2$ is reachable from $\mathcal{C}_1$ if for every $i\in\mathcal{C}_2$, there exists a directed path in $\mathcal{G}$ from some $j\in\mathcal{C}_1$ to agent $i$ (note that $j$ will in general be a function of $i$). 
\begin{theorem}
Let the communication graph be time-invariant and be denoted by $\mathcal{G}$. Suppose the following conditions hold. (i) The observation model satisfies the global identifiability condition (Assumption \ref{ass:identifiability}). (ii) For every pair of hypotheses $\theta_p \neq \theta_q\in\Theta$, the set $\mathcal{V}\setminus\mathcal{S}(\theta_p,\theta_q)$ is reachable from the set $\mathcal{S}(\theta_p,\theta_q)$ in $\mathcal{G}$. Then, Algorithm \ref{Algo:Algo1} guarantees consistency as in Theorem \ref{thm:main}. Furthermore, for every $\theta\in\Theta\setminus\{\theta^{\star}\}$, the following holds for each agent $i\in\mathcal{V}$:
    \begin{equation}
        \liminf_{t\to\infty}-\frac{\log\mu_{i,t}(\theta)}{t} \geq \max_{v\in\mathcal{S}_i(\theta^{\star},\theta)}K_v(\theta^{\star},\theta) \hspace{1mm} a.s.,
        \label{eqn:rate2}
    \end{equation}
where $\mathcal{S}_i(\theta^{\star},\theta)\subseteq\mathcal{S}(\theta^{\star},\theta)$ are those source agents from which there exists a directed path to $i$ in $\mathcal{G}$. 
\label{thm:minimal}
\end{theorem}
\begin{proof}
Fix $\theta\in\Theta\setminus\{\theta^{\star}\}$, and consider an agent $i\in\mathcal{V}\setminus\mathcal{S}(\theta^{\star},\theta)$. The sets $\mathcal{S}(\theta^{\star},\theta)$ and $\mathcal{S}_i(\theta^{\star},\theta)$ are non-empty based on conditions (i) and (ii) of the theorem, respectively. Following a similar line of argument as in the proof of Theorem \ref{thm:main}, one can establish  the following for each $v\in\mathcal{S}_i(\theta^{\star},\theta)$.
\begin{equation}
        \liminf_{t\to\infty}-\frac{\log\mu_{i,t}(\theta)}{t} \geq K_v(\theta^{\star},\theta) \hspace{1mm} a.s.
    \end{equation}
The assertion regarding equation \eqref{eqn:rate2} then follows readily. Consistency follows by noting that since $\mathcal{S}_i(\theta^{\star},\theta)\subseteq\mathcal{S}(\theta^{\star},\theta)$, $K_v(\theta^{\star},\theta) >0, \forall v\in\mathcal{S}_i(\theta^{\star},\theta).$
\end{proof}

Our next result reveals that the combination of conditions (i) and (ii) in Theorem \ref{thm:minimal} constitutes \textit{minimal} requirements on the observation model and the network structure for \textit{any} learning algorithm to guarantee consistency, when the observations of the agents are conditionally independent. 

\begin{theorem}
Let the communication graph be time-invariant and be denoted by $\mathcal{G}$. Then, the following assertions hold.
\begin{enumerate}
    \item[(i)] Conditions (i) and (ii) in Theorem \ref{thm:minimal}, taken together, is equivalent to global identifiability of each source component of $\mathcal{G}$.\footnote{A source component of a time-invariant graph $\mathcal{G}$ is a strongly connected component with no incoming edges.} 
    \item[(ii)] Suppose the observations of the agents are independent conditional on the realization of any state, i.e., $l(\cdot|\theta)=\prod\limits_{i=1}^{n}l_i(\cdot|\theta), \forall \theta\in\Theta$. Then, global identifiability of each source component of $\mathcal{G}$ is necessary and sufficient for unique identification of any true state that gets realized, at every agent, with probability 1.  
\end{enumerate}
\label{thm:neccsuff}
\end{theorem}

The proof of the above result is fairly straightforward and hence omitted here. We now leverage the above results to quantify the rate at which the overall network uncertainty about the true state decays to zero. To measure such uncertainty, we employ the following metric from \cite{jad2} which captures the total variation distance between the agents' beliefs at time-step $t$, and the probability distribution that is concentrated entirely on the true state of the world, namely $\mathbf{1}_{\theta^{\star}}(\cdot)$:
\begin{equation}
    e_t(\theta^{\star})\triangleq\frac{1}{2}\sum\limits^{n}_{i=1}{\Vert\boldsymbol{\mu}_{i,t}(\cdot)-\mathbf{1}_{\theta^{\star}}(\cdot)\Vert}_{1}=\sum\limits^{n}_{i=1}\sum\limits_{\theta\neq\theta^{\star}}\mu_{i,t}(\theta).
    \label{eqn:TV_err}
\end{equation}
Given that $\theta^{\star}$ gets realized, the \textit{rate of social learning} is then defined as \cite{jad2,lalitha1}:
\begin{equation}
    \rho_L(\theta^{\star})\triangleq \liminf_{t\to\infty} -\frac{1}{t}\log e_t(\theta^{\star}).
\end{equation}
Notice that the above expression depends on the state being realized; to account for the realization of any state, one can simply look at the quantity $\min_{\theta^{\star}\in\Theta}\rho_L(\theta^{\star})$ that provides a sense for the least rate of learning one can expect given a certain observation model, a network, and a consistent learning algorithm. We have the following immediate corollaries of Theorems \ref{thm:main} and \ref{thm:minimal}; their proofs are trivial and hence omitted.

\begin{corollary}
Suppose the conditions stated in Theorem \ref{thm:main} are met. Then, Algorithm \ref{Algo:Algo1} guarantees:
\begin{equation}
    \rho_L(\theta^{\star}) \geq \min_{\theta\neq\theta^{\star}} \max_{v\in\mathcal{S}(\theta^{\star},\theta)}K_v(\theta^{\star},\theta) \hspace{1mm} a.s. 
    \label{eqn:rate_social1}
\end{equation}
\label{corr_main}
\end{corollary}

\begin{corollary}
Suppose the conditions stated in Theorem \ref{thm:minimal} are met. Then, Algorithm \ref{Algo:Algo1} guarantees:
\begin{equation}
    \rho_L(\theta^{\star}) \geq \min_{\theta\neq\theta^{\star}} \min_{i\in\mathcal{V}} \max_{v\in\mathcal{S}_i(\theta^{\star},\theta)}K_v(\theta^{\star},\theta) \hspace{1mm} a.s. 
    \label{eqn:rate_social2}    
\end{equation}
\label{corr_minimal}
\end{corollary}
\subsection{Discussion of the Results}
\label{sec:discussion}
\noindent \textbf{Comments on Theorem \ref{thm:main}}: Let us compare the rate of learning based on our method to those existing in literature. Under identical assumptions of global identifiability of the observation model, and strong-connectivity (or joint strong-connectivity as in \cite{nedic}) of the underlying communication graph, both linear \cite{GEBjad,jad2} and log-linear \cite{shahinTAC,lalitha1,nedic} opinion pooling lead to an asymptotic rate of rejection of the form $\sum_{i\in\mathcal{V}}\nu_iK_i(\theta^{\star},\theta)$ for each false hypothesis $\theta\in\Theta\setminus\{\theta^{\star}\}$, for each agent $i\in\mathcal{V}$.\footnote{In \cite{nedic}, the consensus weights are chosen to obtain a network-structure independent (albeit network-size dependent) rate of rejection of $\theta$ of the form $\frac{1}{n}\sum_{i\in\mathcal{V}}K_i(\theta^{\star},\theta)$.} Here, $\nu_i$ represents the eigenvector centrality of agent $i\in\mathcal{V}$, which is strictly positive for a strongly-connected graph. 
Thus, referring to equation \eqref{eqn:asymprate}  reveals that the asymptotic rate of rejection of each false hypothesis (and hence, the rate of social learning) resulting from our algorithm (see \eqref{eqn:rate_social1}), is a \textit{strict improvement} over all existing rates - this constitutes a significant contribution of our paper. Furthermore, observe from Corollary \ref{corr_main} that the lower bound on the rate of social learning is \textit{independent of both the size and structure of the network}. A key implication of this result is the fact that as long as the total information content of the network remains the same, the specific manner in which signals are allocated to agents does not impact the long-run learning rate of our approach. In sharp contrast, existing learning rates that depend on the agents' eigenvector centralities may suffer under poor signal allocations; see \cite{jad2} for a  discussion on this topic.

\noindent \textbf{Comments on Theorem \ref{thm:conc}}: At any given time $t$, for some $i\in\mathcal{V}$ and $\theta\neq\theta^{\star}$, let us consider the set of all sample paths where agent $i$'s instantaneous rate of rejection of $\theta$ is lower than its asymptotic lower bound by a constant additive factor of $\epsilon$. Theorem \ref{thm:conc} complements Theorem \ref{thm:main} by telling us that an arbitrarily accurate approximation of the measure of such ``bad" sample paths eventually decays to zero at an exponential rate no smaller than $\epsilon^2/8L^2$ (the approximation is arbitrarily accurate since the set $\Omega'(\delta)$ can be chosen to have measure arbitrarily close to $1$). It is instructive to compare the concentration result of Theorem \ref{thm:conc} with \cite[Theorem 2]{nedic}, \cite[Theorem 2]{lalitha1}, and \cite[Lemma 3]{shahinTAC}. The analogous results in these papers are more elegant relative to ours, since they do not involve a set of the form $\Omega'(\delta)$ that shows up in our analysis. A refinement of Theorem \ref{thm:conc} to obtain a cleaner non-asymptotic result would require a precise characterization of the transient dynamics generated by our learning rule: we reserve investigations along this line as future work.

\noindent \textbf{Comments on Theorem \ref{thm:minimal}}: While Theorem \ref{thm:neccsuff} identifies an \textit{algorithm-independent} necessary condition for ensuring unique identifiability of any realized state at every agent (when the communication graph is time-invariant and agents receive conditionally independent signals),  Theorem \ref{thm:minimal} reveals that such a condition is also sufficient for our proposed learning algorithm to work. We believe that a result of this flavor is missing in the existing literature on distributed hypothesis testing, where strong-connectivity is a standard assumption. The authors in \cite{molavi2013} do relax the strong-connectivity assumption, but require \textit{every} strongly-connected component of $\mathcal{G}$ to be globally identifiable for learning to take place \cite[Proposition 4]{molavi2013}. In contrast, Theorem \ref{thm:minimal} requires only the source components of $\mathcal{G}$ to satisfy the global identifiability requirement. Interestingly, our conclusions in this context align with an analogous result that identifies joint detectability of each source component as the minimal requirement for solving the related problem of distributed state estimation \cite{martins,mitraTAC}. 

The more general network condition in Theorem \ref{thm:minimal} (as opposed to strong-connectivity) comes at the cost of a potential reduction in the rate of social learning, as reflected in Corollary \ref{corr_minimal}. When the underlying graph is strongly-connected, $\mathcal{S}_i(\theta^{\star},\theta)=\mathcal{S}(\theta^{\star},\theta)$. Consequently, the min w.r.t. the agent set $\mathcal{V}$ in equation \eqref{eqn:rate_social2} goes away, and we recover Corollary \ref{corr_main}.

\section{Learning despite Misinformation}
\label{sec:LFRHE}
In this section, we will address the problem of learning the true state of the world despite the presence of certain agents who do not behave as expected and deliberately try to spread misinformation. In order to isolate the challenges introduced by such malicious entities, we will consider a time-invariant communication graph $\mathcal{G}$ for our subsequent discussion; we anticipate that our proposed approach will extend to the time-varying case with suitable modifications. We now describe the model of agent-misbehavior that we consider.\footnote{Different from our setting, the \textit{forceful
}agents in \cite{acemoglu} do not behave arbitrarily and, in fact, update their beliefs (even if infrequently) by interacting with their neighbors; our adversary model makes no such assumptions.}

\textbf{Adversary Model:} We assume that a certain subset of the agents are adversarial, and model their behavior based on the Byzantine fault model \cite{Byz}. Specifically, Byzantine agents possess complete knowledge of the observation model, the network model, the algorithms being used, the information being exchanged, and the true state of the world. Leveraging such information, adversarial agents can behave arbitrarily and in a coordinated manner, and can in particular, send incorrect, potentially inconsistent information to their out-neighbors. In return for allowing such worst-case adversarial behavior and knowledge by the adversaries, we will restrict the number of such adversaries; in particular, we will consider an $f$-local adversarial model, i.e., we assume that there are at most $f$ adversaries in the neighborhood of any non-adversarial agent, where $f\in\mathbb{N}$.  Finally, we emphasize that the non-adversarial agents are unaware of the identities of the adversaries in their neighborhood. As is fairly standard in the distributed fault-tolerant literature \cite{vaidyacons,rescons,dibaji,suopt,sundaramopt,mitraarxiv,usevitch1,broad1}, we only assume that non-adversarial agents know the upper bound $f$ on the number of adversaries in their neighborhood. The adversarial set will be denoted by $\mathcal{A}\subset\mathcal{V}$, and the remaining agents $\mathcal{R}=\mathcal{V}\setminus\mathcal{A}$ will be called the regular agents. 

Our immediate goals are as follows. (i) Devise an algorithm that enables each regular agent to asymptotically identify the true state with probability $1$, despite the presence of an $f$-local Byzantine adversarial set. (ii) Identify conditions on the observation model and the network structure that guarantee correctness of such an algorithm. Prior to addressing these goals, we briefly motivate the need for a novel Byzantine-resilient learning algorithm.

\textbf{Motivation}: A standard way to analyze the impact of adversarial agents while designing resilient distributed consensus-based protocols (for applications in consensus \cite{rescons,vaidyacons}, optimization \cite{suopt,sundaramopt}, hypothesis testing \cite{su1}, and multi-agent rendezvous \cite{seth1}) is to construct an equivalent matrix representation of the linear update rule that involves only the regular agents \cite{vaidyamatrix}. In particular, this requires expressing the iterates of a regular agent as a convex combination of the iterates of its regular neighbors, based on appropriate filtering techniques, and under certain assumptions on the network structure. While this can indeed be achieved efficiently for scalar consensus problems, for problems requiring consensus on vectors (like the belief vectors in our setting), such an approach typically requires the computation of sets known as  \textit{Tverberg partitions}. However, there is no known algorithm that can compute an exact Tverberg partition in polynomial time for a general $d$-dimensional finite point set \cite{tver}. Consequently, since the filtering approach developed in \cite{su1} requires each regular agent to compute a Tverberg partition at every iteration, the resulting computations are forbiddingly high. The authors in \cite{su1} do briefly discuss an alternate pairwise learning rule that requires agents to perform scalar consensus on relative confidence levels (instead of beliefs) of one hypothesis over another. Under such a rule, for each regular agent, its relative confidence on the true state over every false hypothesis approaches infinity - a condition that is difficult to verify in practice. Moreover, the pairwise learning rule in \cite{su1} requires each agent to maintain and update at each time-step a vector of dimension $O(m^2)$. In contrast, we propose a simple, light-weight Byzantine-resilient learning rule that avoids the computation of Tverberg partitions, and requires agents to update two $m$-dimensional belief vectors.
\begin{algorithm}[t]
\caption{Belief update rule for each $i\in\mathcal{R}$}
\label{Algo:Algo2}
\begin{algorithmic}[1]
\State \textbf{Initialization:}  $\mu_{i,0}(\theta)>0$,  $\pi_{i,0}(\theta)>0$, $\forall \theta\in\Theta$, and  $\sum_{\theta\in\Theta}\mu_{i,0}(\theta)=1$, $\sum_{\theta\in\Theta}\pi_{i,0}(\theta)=1$
 
\State Transmit  $\boldsymbol{\mu}_{i,0}$ to  out-neighbors
  
\For{$t+1\in\mathbb{N}_{+}$}
\For{$\theta\in\Theta$}
   
\State  Update local belief on $\theta$ as per \eqref{eqn:Bayes}

\If{$|\mathcal{N}_i| \geq (2f+1)$}
\State Sort   $\mu_{j,t}(\theta),j\in\mathcal{N}_i$ from highest to lowest, and reject the highest $f$ and the lowest $f$ of such beliefs.
             
\State Let $\mathcal{M}^{\theta}_{i,t}$ be the set of agents whose beliefs are not rejected in the previous step. Update $\mu_{i,t+1}(\theta)$ as 
\begin{equation}
\mu_{i,t+1}(\theta)=\frac{\min\{\{\mu_{j,t}(\theta)\}_{j\in\mathcal{M}^{\theta}_{i,t}},\pi_{i,t+1}(\theta)\}}{\sum\limits_{p=1}^{m}\min\{\{\mu_{j,t}(\theta_p)\}_{j\in\mathcal{M}^{\theta_p}_{i,t}},\pi_{i,t+1}(\theta_p)\}}
\label{eqn:rule2}
\end{equation}
\Else
\State Update $\mu_{i,t+1}(\theta)$ as
\begin{equation}
\mu_{i,t+1}(\theta)=\pi_{i,t+1}(\theta)
\label{eqn:rule3}
\end{equation}
\EndIf
\EndFor
\State Transmit $\boldsymbol{\mu}_{i,t+1}$ to out-neighbors
\EndFor
\end{algorithmic}
\end{algorithm}

\subsection{A Byzantine-Resilient Distributed Learning Rule}
In this section, we develop an easy to implement and computationally-efficient extension of Algorithm \ref{Algo:Algo1} that guarantees learning despite the presence of Byzantine adversaries. We call it the Local-Filtering based Resilient Hypothesis Elimination (LFRHE) algorithm (Algorithm \ref{Algo:Algo2}). Like Algorithm \ref{Algo:Algo1}, the LFRHE algorithm requires every regular agent $i$
to maintain and update (at every time-step $t$) a local belief vector $\boldsymbol{\pi}_{i,t}$, and an actual belief vector $\boldsymbol{\mu}_{i,t}$. While $\boldsymbol{\pi}_{i,t}$ is updated as before via \eqref{eqn:Bayes}, the update of $\boldsymbol{\mu}_{i,t}$ is the key feature of Algorithm \ref{Algo:Algo2}. To update $\mu_{i,t+1}(\theta)$, agent $i\in\mathcal{R}$ first checks whether it has at least $2f+1$ neighbors. If it does, then it rejects the highest $f$ and the lowest $f$ neighboring beliefs $\mu_{j,t}(\theta),j \in \mathcal{N}_i$ (line 7 in Algo. \ref{Algo:Algo2}), and employs a min-rule as before, but using only the remaining beliefs (line 8 in Algo. \ref{Algo:Algo2}). Thus, agent $i$
filters out the most extreme neighboring beliefs on each hypothesis, and retains only the moderate ones to update its own actual belief. If agent $i$ has strictly fewer than $2f+1$ neighbors, then it decides against using neighboring information and, instead, updates its actual belief vector to be equal to its local belief vector (line 10 in Algo. \ref{Algo:Algo2}).

To state our main result concerning the correctness of Algorithm \ref{Algo:Algo2}, we require the following definitions. 
\begin{definition}($r$-\textbf{reachable set}) \cite{rescons} For a graph $\mathcal{G}=(\mathcal{V,E})$, a set $\mathcal{C} \subseteq \mathcal{V}$, and an integer $r \in \mathbb{N}_{+}$, $\mathcal{C}$ is an \textit{$r$-reachable set} if there exists an $i \in \mathcal{C}$ such that $|\mathcal{N}_i \setminus \mathcal{C}| \geq r$.
\label{defn:rreachable}
\end{definition}
\begin{definition}(\textbf{strongly} $r$-\textbf{robust graph} \textit{w.r.t.} $\mathcal{S}(\theta_p,\theta_q)$) For $r \in \mathbb{N}_{+}$ and $\theta_p,\theta_q\in\Theta$, a graph $\mathcal{G}=(\mathcal{V,E})$ is \textit{strongly $r$-robust w.r.t. the set of source agents $\mathcal{S}(\theta_p,\theta_q)$}, if for every non-empty subset $\mathcal{C} \subseteq \mathcal{V}\setminus\mathcal{S}(\theta_p,\theta_q)$, $\mathcal{C}$ is $r$-reachable.
\label{defn:strongrobust}
\end{definition}
\begin{theorem}
Suppose that for every pair of hypotheses $\theta_p,\theta_q\in\Theta$, the graph $\mathcal{G}$ is strongly $(2f+1)$-robust w.r.t. the  source set $\mathcal{S}(\theta_p,\theta_q)$.
Then, Algorithm \ref{Algo:Algo2} guarantees the following despite the actions of any $f$-local set of Byzantine adversaries.
\begin{itemize}
    \item \textbf{(Consistency)}: For each agent $i\in\mathcal{R}$,  $\mu_{i,t}(\theta^{\star}) \rightarrow 1$ a.s.
    \item \textbf{(Asymptotic Rate of Rejection of False Hypotheses)}: Consider any false hypothesis $\theta\in\Theta\setminus\{\theta^{\star}\}$. Then, the following holds for each agent $i\in\mathcal{R}$.
    \begin{equation}
        \liminf_{t\to\infty}-\frac{\log\mu_{i,t}(\theta)}{t} \geq \min_{v\in\mathcal{S}(\theta^{\star},\theta)\cap\mathcal{R}}K_v(\theta^{\star},\theta) \hspace{1mm} a.s.
        \label{eqn:rate_Byz}
    \end{equation}
\end{itemize}
\label{thm:Byz} 
\end{theorem}

\begin{proof} See Appendix \ref{sec:App3}. \end{proof} 

\begin{remark} 
For any pair $\theta_p,\theta_q\in\Theta$, notice that the strong-robustness condition in Theorem \ref{thm:Byz} (together with Def. \ref{defn:strongrobust}) requires $|\mathcal{S}(\theta_p,\theta_q)|\geq(2f+1)$, if $\mathcal{V}\setminus\mathcal{S}(\theta_p,\theta_q)$ is non-empty. In particular, it blends requirements on the signal structures of the agents with those on the communication graph. To gain intuition about this condition, suppose $\Theta=\{\theta_1,\theta_2\}$, and consider an agent $i\in\mathcal{V}\setminus\mathcal{S}(\theta_1,\theta_2)$. To enable $i$ to learn the truth despite potential adversaries in its neighborhood, one requires (i) redundancy in the signal structures of the agents, and (ii) redundancy in the network structure to ensure reliable information flow from $\mathcal{S}(\theta_1,\theta_2)$ to agent $i$. These requirements are encapsulated by Theorem \ref{thm:Byz}. For a fixed source set $\mathcal{S}(\theta_p,\theta_q)$, checking whether $\mathcal{G}$ is strongly $(2f+1)$-robust w.r.t. $\mathcal{S}(\theta_p,\theta_q)$ can be done in polynomial time by drawing connections to the process of bootstrap percolation on networks \cite[Proposition 5]{mitraarxiv}. Since the source sets for each pair $\theta_p,\theta_q\in\Theta$ can also be computed in polynomial time via a simple inspection of the agents' signal structures, it follows that the strong-robustness condition in  Theorem \ref{thm:Byz} can be checked in polynomial time.
\end{remark}

Leveraging Theorem \ref{thm:Byz}, we can characterize the rate of decay of the collective uncertainty of the regular agents regarding the true state. To do so, we employ the following modification of the metric \eqref{eqn:TV_err}:
\begin{equation}
    e^{\mathcal{R}}_t(\theta^{\star})\triangleq\frac{1}{2}\sum\limits_{i\in\mathcal{R}}{\Vert\boldsymbol{\mu}_{i,t}(\cdot)-\mathbf{1}_{\theta^{\star}}(\cdot)\Vert}_{1}=\sum\limits_{i\in\mathcal{R}}\sum\limits_{\theta\neq\theta^{\star}}\mu_{i,t}(\theta).
    \label{eqn:TV_err_Byz}
\end{equation}
Note that this metric only considers the beliefs of the regular agents, as the Byzantine agents can update their beliefs however they wish. 
With $\theta^{\star}$ as the true state, we define the rate of social learning in the presence of Byzantine adversaries as:
\begin{equation}
    \rho^{\mathcal{R}}_L(\theta^{\star})\triangleq \liminf_{t\to\infty} -\frac{1}{t}\log e^{\mathcal{R}}_t(\theta^{\star}).
\end{equation}
We have the following immediate corollary of Theorem \ref{thm:Byz}.
\begin{corollary}
Suppose the conditions stated in Theorem \ref{thm:Byz} are met. Then, Algorithm \ref{Algo:Algo2} guarantees:
\begin{equation}
    \rho^{\mathcal{R}}_L(\theta^{\star}) \geq \min_{\theta\neq\theta^{\star}} \min_{v\in\mathcal{S}(\theta^{\star},\theta)\cap\mathcal{R}}K_v(\theta^{\star},\theta) \hspace{1mm} a.s. 
\label{eqn:corr_Byz}
\end{equation}
\end{corollary}
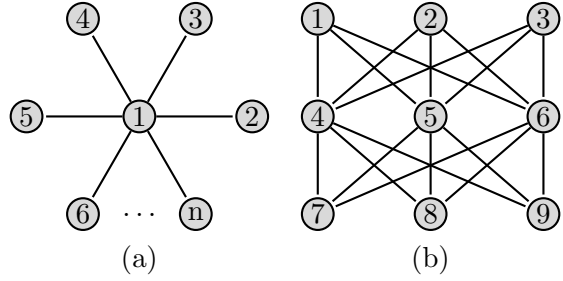
\begin{figure}[t]
\begin{center}
\begin{tabular}{c c}
\begin{tikzpicture}
[->,shorten >=1pt,scale=.75,inner sep=1pt, minimum size=12pt, auto=center, node distance=3cm,
  thick, node/.style={circle, draw=black, thick},]
\node [circle, draw, fill=gray!30](n1) at (0,0)  (1)  {1};
\node [circle, draw, fill=gray!30](n2) at (2,0)   (2)  {2};
\node [circle, draw, fill=gray!30](n3) at (1,1.7321)   (3)  {3};
\node [circle, draw, fill=gray!30](n4) at (-1,1.7321)  (4)  {4};
\node [circle, draw, fill=gray!30](n5) at (-2,0)   (5)  {5};   
\node [circle, draw, fill=gray!30](n6) at (-1,-1.7321)   (6)  {6}; 
\node [circle, draw, fill=gray!30](n7) at (1,-1.7321)   (n)  {n}; 
\node [black] at (0,-1.732) () {$\ldots$};
\draw [-]
(1) -- (2); 
\draw [-]
(1) -- (3); 
\draw [-]
(1) -- (4); 
\draw [-]
(1) -- (5); 
\draw [-]
(1) -- (6); 
\draw [-]
(1) -- (n); 
\end{tikzpicture}
&
\begin{tikzpicture}
[->,shorten >=1pt,scale=.75,inner sep=1pt, minimum size=12pt, auto=center, node distance=3cm,
  thick, node/.style={circle, draw=black, thick},]
\node [circle, draw, fill=gray!30](n1) at (-2,1.732)  (1)  {1};
\node [circle, draw, fill=gray!30](n2) at (0,1.732)   (2)  {2};
\node [circle, draw, fill=gray!30](n3) at (2,1.732)   (3)  {3};
\node [circle, draw, fill=gray!30](n4) at (-2,0)  (4)  {4};
\node [circle, draw, fill=gray!30](n5) at (0,0)   (5)  {5};
\node [circle, draw, fill=gray!30] (n6) at (2,0)   (6)  {6};
\node [circle, draw, fill=gray!30](n7) at (-2,-1.732)   (7)  {7};
\node [circle, draw, fill=gray!30] (n8) at (0,-1.732)  (8)  {8};
\node [circle, draw, fill=gray!30] (n9) at (2,-1.732)  (9)  {9};
\draw[-]
 (1) -- (4);
\draw[-]
 (1) -- (5);
 \draw[-]
 (1) -- (6);
 \draw[-]
 (2) -- (4);
 \draw[-]
 (2) -- (5);
 \draw[-]
 (2) -- (6);
 \draw[-]
 (3) -- (4);
 \draw[-]
 (3) -- (5);
 \draw[-]
 (3) -- (6);
 \draw[-]
 (4) -- (7);
\draw[-]
 (4) -- (8);
 \draw[-]
 (4) -- (9);
 \draw[-]
 (5) -- (7);
 \draw[-]
 (5) -- (8);
 \draw[-]
 (5) -- (9);
 \draw[-]
 (6) -- (7);
 \draw[-]
 (6) -- (8);
 \draw[-]
 (6) -- (9);
\end{tikzpicture}\\
(a) & (b)
\end{tabular}
\end{center}
\caption{Figures \ref{fig:networks}(a) and \ref{fig:networks}(b) represent the network models for simulation examples 1 and 2, respectively.}
\label{fig:networks}
\end{figure}
\begin{figure}[t]
\begin{tabular}{c c}
\includegraphics[scale=0.35]{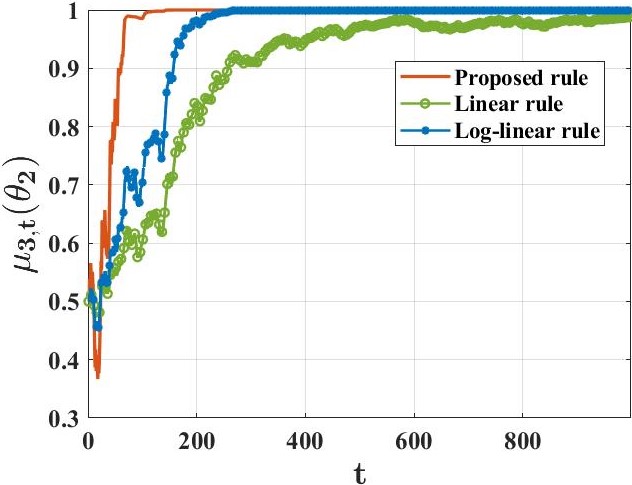}&\hspace{-3mm}\includegraphics[scale=0.35]{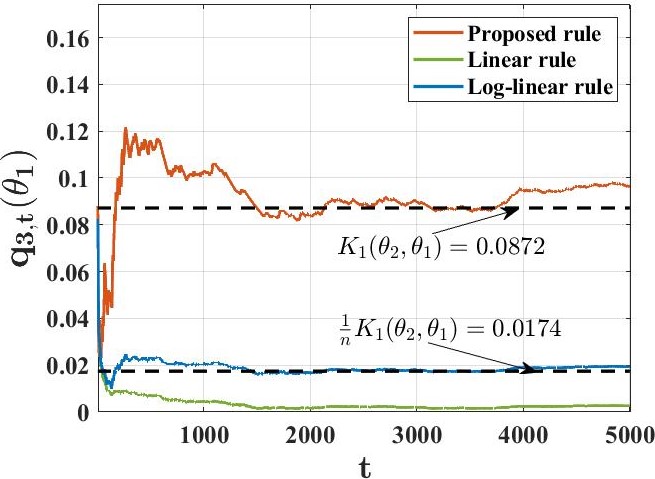}\\
(a)&(b)
\end{tabular}
\caption{Consider the setup of simulation  example 1 with $n=5$ agents. Fig. \ref{fig:case1}(a) depicts the evolution of agent 3's belief on the true state $\theta_2$, and Fig. \ref{fig:case1}(b) depicts the evolution of the instantaneous rate of rejection of $\theta_1$ for agent 3, namely $q_{3,t}(\theta_1)=-\log\mu_{3,t}(\theta_1)/{t}$.}
\label{fig:case1}
\end{figure}
\begin{figure}[t]
\begin{tabular}{c c}
\includegraphics[scale=0.35]{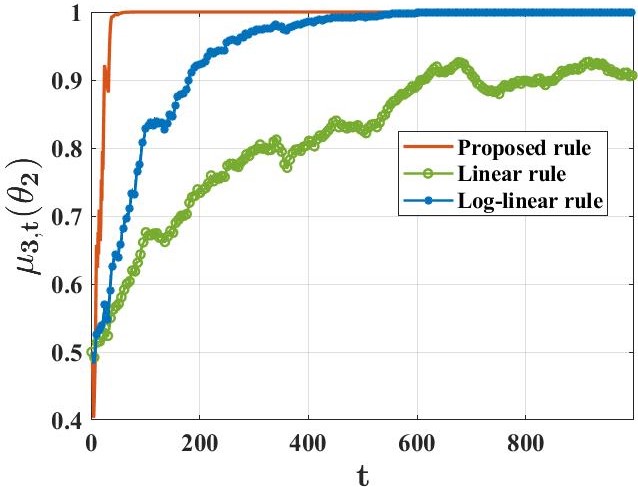}&\hspace{-3mm}\includegraphics[scale=0.35]{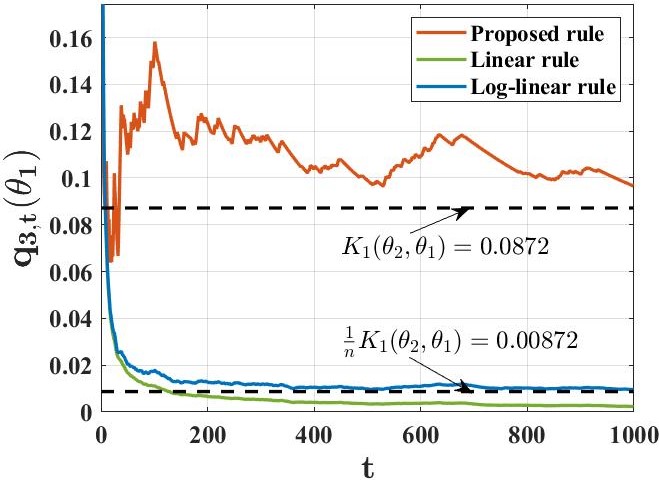}\\
(a)&(b)
\end{tabular}
\caption{Consider the setup of simulation example 1 with $n=10$ agents. Fig. \ref{fig:case2} illustrates the dilution in the rates of social learning for the linear and log-linear rules with an increase in the number of uninformative agents. Figures \ref{fig:case2}(a) and \ref{fig:case2}(b) are analogous to those in Figure \ref{fig:case1}.}
\label{fig:case2}
\end{figure}
\begin{figure}[hbt!]
\begin{tabular}{c c}
\includegraphics[scale=0.35]{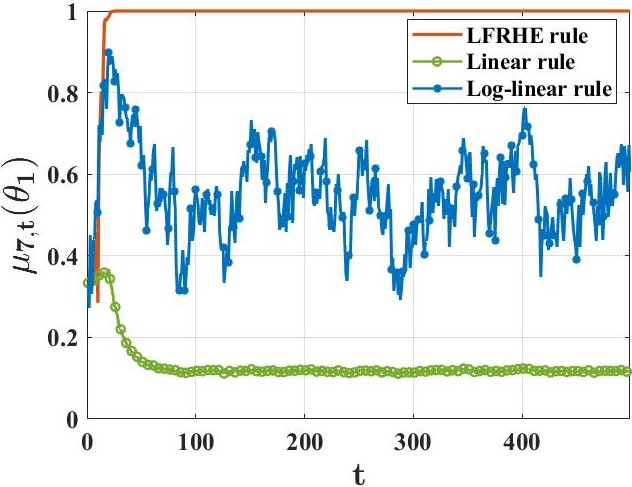}&\hspace{-3mm}\includegraphics[scale=0.35]{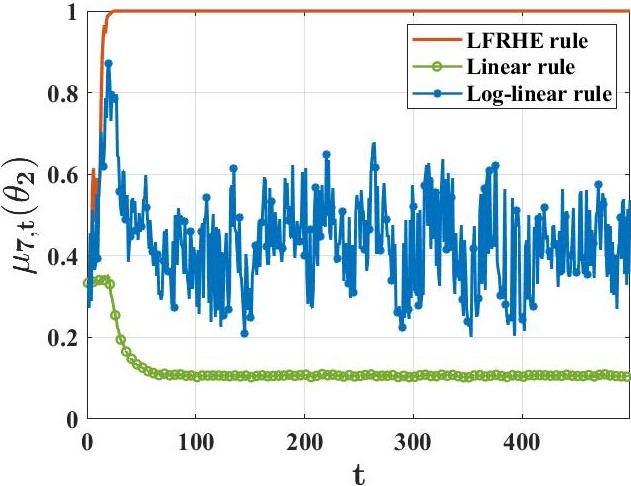}\\
(a)&(b)
\end{tabular}
\caption{Consider the setup of simulation example 2, where agent 5 acts as an adversary. Figures \ref{fig:case2}(a) and \ref{fig:case2}(b) depict the evolution of agent 7's belief on the true state, when $\theta^{\star}=\theta_1$, and $\theta^{\star}=\theta_2$, respectively.}
\label{fig:Byz}
\end{figure}
\section{Simulations}
\noindent \textbf{Example 1 (Impact of Network Size on Rate of Convergence):} For our first simulation study, we consider a binary hypothesis testing problem, i.e., $\Theta=\{\theta_1,\theta_2\}$, where the signal space for each agent is identical and comprises of signals $w_1$ and $w_2$. The (time-invariant) undirected network for this example is depicted in Figure \ref{fig:networks}(a). The likelihood models of the agents are as follows:  $l_1(w_1|\theta_1)=0.7, l_1(w_1|\theta_2)=0.5$, and $l_i(w_1|\theta_1)=l_i(w_1|\theta_2)=0.5, \forall i\in\mathcal{V}\setminus\{1\}$, i.e., agent 1 is the only informative agent. In order to compare the performance of Algorithm \ref{Algo:Algo1} to the linear and log-linear belief update rules in \cite{GEBjad} and \cite{nedic}, we implement the latter assuming consensus weights are assigned based on the lazy Metropolis scheme (see \cite{nedic} for details). Based on this weight assignment, it is easy to verify that the eigenvector centrality of each agent is $1/n$. All agents start out with uniform priors. With $\theta^{\star}=\theta_2$, and $n=5$, Figure \ref{fig:case1} illustrates the performance of the three algorithms w.r.t. agent 3. In particular, Figure \ref{fig:case1}(a) reveals that based on our approach, agent 3's belief on the true state $\theta_2$ converges to 1 faster than the other algorithms. Figure \ref{fig:case1}(b) makes this observation  precise by plotting the instantaneous rate of rejection of $\theta_1$ for agent 3, namely $q_{3,t}(\theta_1)=-\log\mu_{3,t}(\theta_1)/{t}$. Consistent with the respective theoretical findings, $q_{3,t}(\theta_1)$ is eventually lower-bounded by $K_1(\theta_2,\theta_1)$ for our algorithm (see Theorem \ref{thm:main}), approaches $K_1(\theta_2,\theta_1)/n$ for the log-linear rule in  \cite{nedic}, and is eventually upper-bounded by $K_1(\theta_2,\theta_1)/n$ for the linear rule in \cite{GEBjad}. Similar conclusions hold for the other agents. 

Suppose we now double the number of agents in the network. Agent 1 continues to remain the only informative agent. Figure \ref{fig:case2} compares the performances of the three algorithms for this case. Notably, the convergence rate for our approach remains unaffected, whereas that for the linear and log-linear rules gets diluted. This observation can be attributed to the fact that while the rate provided by our algorithm is both network-structure and network-size independent for strongly-connected networks (see Section \ref{sec:discussion}), the rates of the linear and log-linear rules depend crucially on the eigenvector centralities of the agents, which, in this case, correspond to $1/n$. Thus, the gap between the performance of our algorithm, and that of the linear and log-linear update rules (as measured by convergence rates), becomes more pronounced as the number of uninformative agents increase (i.e., as $n$ increases, but the total information content of the network remains the same). 

\noindent \textbf{Example 2 (Impact of Adversaries):} While the previous example highlighted the benefits of Algorithm \ref{Algo:Algo1}, we now focus on an example that demonstrates the resilience of its variant, namely the LFRHE algorithm (Algorithm \ref{Algo:Algo2}), to the presence of Byzantine adversaries. To this end, consider the undirected network in Figure \ref{fig:networks}(b). For this example, $\Theta=\{\theta_1,\theta_2,\theta_3\}$, and $\mathcal{S}_i=\{w_1,w_2\}, \forall i \in \mathcal{V}$. Suppose the agent likelihood models are given by $l_i(w_1|\theta_1)=3/4, l_i(w_1|\theta_2)=l_i(w_1|\theta_3)=1/3, \forall i\in\{1,2,3\}$, $l_i(w_1|\theta_1)=l_i(w_1|\theta_2)=2/5, l_i(w_1|\theta_3)=1/7, \forall i\in\{4,5,6\}$, and $l_i(w_1|\theta_1)=l_i(w_1|\theta_2)=1/2, l_i(w_1|\theta_3)=5/6, \forall i\in\{7,8,9\}$. Suppose $f=1$ and agent 5 is the only adversarial agent. It is easy to see that condition (i) in Theorem \ref{thm:Byz} is met. We will compare the performance of Algorithm \ref{Algo:Algo2} with the linear rule in \cite{GEBjad}, and the log-linear rule in \cite{nedic}. For implementing the latter, we again assign consensus weights based on the lazy Metropolis scheme. All agents start out with uniform priors. The adversary, agent 5, maintains a belief of $0.1$ on the true state, and $0.45$ on each of the false hypotheses, for all $t\geq 20$. Figures \ref{fig:Byz}(a) and  \ref{fig:Byz}(b) illustrate the repercussions of this action on agent 7, when $\theta^{\star}=\theta_1$ and $\theta^{\star}=\theta_2$, respectively: while the linear and log-linear rules fail to recover from the attack, Algorithm \ref{Algo:Algo2} enables agent 7 to infer the truth. Similar conclusions hold for the other regular agents.
\section{Conclusion}
We proposed and analyzed a novel algorithm for addressing the problem of distributed hypothesis testing. The key distinguishing feature of our learning algorithm is that it does not employ any linear consensus-based data aggregation protocol. Instead, it relies on a ``min-rule" to spread beliefs through the network. Under mild assumptions of global identifiability and joint strong-connectivity, we established consistency of our learning rule. In particular, we showed that the rate of learning resulting from our approach strictly improves upon all existing rates. For static networks, we established consistency of our algorithm under minimal requirements on the observation model and the network structure. Finally, we proposed a simple and computationally-efficient version of our learning rule that accounts for worst-case adversarial behavior on the part of certain agents in the network. As future work, we plan to investigate the impact of communication constraints on the performance of distributed inference/estimation algorithms. 
\appendix
\section{Proof of Theorem \ref{thm:main}}
\label{sec:App1}
The proof of Theorem \ref{thm:main} is based on several intermediate results. We start with the following simple lemma that characterizes the asymptotic behavior of the local belief sequences generated based on \eqref{eqn:Bayes}; we provide a proof (adapted to our notation) to keep the paper self-contained, and to introduce certain quantities that will be referenced later in our analysis. 
\begin{lemma}
Consider a false hypothesis $\theta\in\Theta\setminus\{\theta^{\star}\}$, and an
agent $i\in\mathcal{S}(\theta^{\star},\theta)$. Suppose $\pi_{i,0}(\theta_p) > 0, \forall \theta_p\in\Theta$. Then, the update rule \eqref{eqn:Bayes} ensures that (i) $\pi_{i,t}(\theta) \rightarrow 0$ a.s., (ii) $\pi_{i,\infty}(\theta^{\star})\triangleq\lim_{t\to\infty}\pi_{i,t}(\theta^{\star})$ exists a.s. and satisfies $\pi_{i,\infty}(\theta^{\star})\geq \pi_{i,0}(\theta^{\star})$, and (iii) the following holds:
\begin{equation}
\lim_{t\to\infty}\frac{1}{t}\log\frac{\pi_{i,t}(\theta)}{\pi_{i,t}(\theta^{\star})}=-K_i(\theta^{\star},\theta) \hspace{1mm} a.s.
\label{eqn:localrate}
\end{equation}
\label{lemma:Bayes}
\end{lemma}
\begin{proof}
Consider any agent $i\in\mathcal{S}(\theta^{\star},\theta)$, and define:
\begin{equation}
\rho^{}_{i,t}(\theta)\triangleq \log\frac{\pi_{i,t}(\theta)}{\pi_{i,t}(\theta^{\star})}, \hspace{2mm} \lambda^{}_{i,t}(\theta) \triangleq \log\frac{l_i(s_{i,t}|\theta)}{l_i(s_{i,t}|\theta^{\star})}.
\end{equation}
Then, based on \eqref{eqn:Bayes}, we obtain the following recursion:
\begin{equation}
\rho_{i,t+1}(\theta)=\rho_{i,t}(\theta)+\lambda_{i,t+1}(\theta), \forall t\in\mathbb{N}.
\label{eqn:recur}
\end{equation}
Rolling out the above equation over time yields
\begin{equation}
\rho_{i,t}(\theta)=\rho_{i,0}(\theta)+\sum \limits_{k=1}^{t}\lambda_{i,k}(\theta), \forall t\in\mathbb{N}_{+}.
\label{eqn:recursion}
\end{equation}
Notice that $\{\lambda_{i,t}(\theta)\}$ is a sequence of i.i.d. random variables with finite means (see equation \eqref{eqn:bounded_lograt}). In particular, it is easy to verify that each random variable $\lambda_{i,t}(\theta)$ has mean\footnote{More precisely, the mean here is obtained by using the expectation operator $\mathbb{E}^{\theta^{\star}}[\cdot]$ associated with the measure $\mathbb{P}^{\theta^{\star}}$.} given by $-K_i(\theta^{\star},\theta)$. Thus, based on the strong law of large numbers, we have $\frac{1}{t}\sum \limits_{k=1}^{t}\lambda_{i,k}(\theta) \rightarrow -K_i(\theta^{\star},\theta)$ almost surely. Dividing both sides of \eqref{eqn:recursion} by $t$, and taking the limit as $t$ goes to infinity, we then obtain
\begin{equation}
\lim_{t\to\infty}\frac{1}{t}\rho_{i,t}(\theta)=-K_i(\theta^{\star},\theta) \hspace{1mm} a.s.,
\label{eqn:limit}
\end{equation}
establishing part (iii) of the lemma. Now note that based on the definition of the set $\mathcal{S}(\theta^{\star},\theta)$, $K_i(\theta^{\star},\theta) > 0$. It then follows from \eqref{eqn:limit} that $\rho_{i,t}(\theta) \rightarrow -\infty$ almost surely, and hence $\pi_{i,t}(\theta) \rightarrow 0$ almost surely. This establishes part (i) of the lemma. For any $\theta\in\Theta^{\theta^{\star}}_i$, observe that $\lambda_{i,t}(\theta)=0, \forall t\in\mathbb{N}_{+}$. It then follows from \eqref{eqn:recur} that for each $\theta\in\Theta^{\theta^{\star}}_i$, $\rho_{i,t}(\theta)=\rho_{i,0}(\theta), \forall t\in\mathbb{N}_{+}$. From the above discussion, we conclude that a limiting belief vector $\boldsymbol{\pi}_{i,\infty}$ exists almost surely, with non-zero entries corresponding to each $\theta\in\Theta^{\theta^{\star}}_i$. Part (ii) of the lemma then follows readily.
\end{proof}

While our proposed learning rule is tailored to facilitate propagation of low beliefs on false hypotheses, it is crucial to also ensure that the beliefs of all agents on the true state remain bounded away from zero. In particular, consider the following scenario. During a transient phase, certain agents see private signals that cause them to temporarily lower their local beliefs on the true state. This effect manifests itself in the actual beliefs of the agents via the min-rule \eqref{eqn:rule1}. We ask: can such a transient phenomenon trigger a cascade of progressively lower beliefs on the true state? The next important result asserts that this will almost surely never be the case.

\begin{lemma}
Suppose the conditions stated in Theorem \ref{thm:main} hold, and Algorithm \ref{Algo:Algo1} is employed by each agent. Then, there exists a set $\bar{\Omega}\subseteq\Omega$ with the following properties: (i) $\mathbb{P}^{\theta^{\star}}(\bar{\Omega})=1$, and (ii) for each $\omega\in\bar{\Omega}$, there exist constants $\eta(\omega)\in(0,1)$ and $t'(\omega)\in(0,\infty)$ such that on the sample path $\omega$,
\begin{equation}
\pi_{i,t}(\theta^{\star}) \geq \eta(\omega), \mu_{i,t}(\theta^{\star}) \geq \eta(\omega), \forall t \geq t'(\omega),\forall i\in\mathcal{V}.
\label{eqn:lowerbound}
    \end{equation}
\label{lemma:lower_bound}
\end{lemma}
\begin{proof} Let $\bar{\Omega}\subseteq\Omega$ denote the set of sample paths for which assertions (i)-(iii) in Lemma \ref{lemma:Bayes} hold for each false hypothesis $\theta\in\Theta\setminus\{\theta^{\star}\}$. Based on Lemma \ref{lemma:Bayes}, we note that $\mathbb{P}^{\theta^{\star}}(\bar{\Omega})=1$. Consequently, to prove the result, it suffices to establish the existence of $\eta(\omega) \in (0,1)$, and $t'(\omega)\in (0,\infty)$ for each sample path $\omega \in \bar{\Omega}$, such that \eqref{eqn:lowerbound} holds. To this end, fix an arbitrary sample path $\omega\in\bar{\Omega}.$ We first argue that the local beliefs of every agent on the true state $\theta^{\star}$ are bounded away from $0$ on $\omega$. To see this, pick any agent $i\in\mathcal{V}$. Suppose there exists some $\theta\in\Theta\setminus\{\theta^{\star}\}$ for which $i\in\mathcal{S}(\theta^{\star},\theta)$. Then, based on our choice of $\omega$,  Lemma \ref{lemma:Bayes} implies that $\pi_{i,\infty}(\theta^{\star})\geq\pi_{i,0}(\theta^{\star})>0$, where the last inequality follows from the requirement of non-zero priors in line 1 of Algo. \ref{Algo:Algo1}. In particular, given the structure of the update rule \eqref{eqn:Bayes}, it follows that $\pi_{i,t}(\theta^{\star}) > 0$ for all time. This is true  
since if $\pi_{i,t}(\theta^{\star})=0$ at any instant, then the corresponding belief would remain at $0$ for all subsequent time-steps, thereby violating the fact that $\pi_{i,\infty}(\theta^{\star})\geq\pi_{i,0}(\theta^{\star})>0$. 
Now consider the scenario where there exists no  $\theta\in\Theta\setminus\{\theta^{\star}\}$ for which $i\in\mathcal{S}(\theta^{\star},\theta)$, i.e.,  every hypothesis in $\Theta$ is observationally equivalent to $\theta^{\star}$ from the point of view of agent $i$. In this case, it is easy to see that based on \eqref{eqn:Bayes}, $\boldsymbol{\pi}_{i,t}=\boldsymbol{\pi}_{i,0}, \forall t\in\mathbb{N}_{+}$. In particular, this implies $\pi_{i,t}(\theta^{\star})=\pi_{i,0}(\theta^{\star})>0, \forall t \in \mathbb{N}_{+}$. This establishes our claim that on $\omega$, $\pi_{i,t}(\theta^{\star})$ remains bounded away from zero $\forall i\in\mathcal{V}$. 

 To proceed, define $\gamma_1\triangleq\min_{i\in\mathcal{V}} \pi_{i,0}(\theta^{\star})>0$, where the inequality follows from line 1 in Algo \ref{Algo:Algo1}. Pick a small number $\delta > 0$ such that $\delta < \gamma_1$, and notice that our discussion concerning the evolution of the local beliefs readily implies the existence of a time-step $t'(\omega)$, such that for all $t \geq t'(\omega)$, $\pi_{i,t}(\theta^{\star}) \geq  \gamma_1-\delta > 0, \forall i\in\mathcal{V}$. With  $\gamma_2(\omega)\triangleq\min_{i\in\mathcal{V}}\{\mu_{i,t'(\omega)}(\theta^{\star})\}$, we claim that $\gamma_2(\omega) > 0$. The claim follows by noting that given the structure of the update rule \eqref{eqn:rule1}, and the requirement of non-zero priors in Algo \ref{Algo:Algo1},  $\gamma_2(\omega)$ can equal $0$
if and only if some agent in the network sets its local belief on $\theta^{\star}$ to $0$ at some time-step prior to $t'(\omega)$. However, this possibility is ruled out in view of the previously established fact that on $\omega$, $\pi_{i,t}(\theta^{\star})>0, \forall t\in\mathbb{N}, \forall i\in\mathcal{V}$.
Let $\eta(\omega)=\min\{\gamma_1-\delta,\gamma_2(\omega)\} > 0$.
In words, $\eta(\omega)$ lower-bounds the lowest belief (considering both local and actual beliefs) on the true state $\theta^{\star}$ held by an agent at time-step $t'(\omega)$. It is apparent from the preceding discussion that $\pi_{i,t}(\theta^{\star})\geq\eta(\omega),\forall t\geq t'(\omega),\forall i\in\mathcal{V}$. Thus, to complete the proof, it remains to establish that $\mu_{i,t}(\theta^{\star})\geq\eta(\omega),\forall t\geq t'(\omega),\forall i\in\mathcal{V}$. To this end, let us fix an agent $i$ and observe the following:
\begin{equation}
\begin{aligned}
\mu_{i,{t}'(\omega)+1}(\theta^{\star})&\overset{(a)}{=}\frac{\min\{\{\mu_{j,t'(\omega)}(\theta^{\star})\}_{{j\in\mathcal{N}_i[t'(\omega)]\cup\{i\}}},\pi_{i,t'(\omega)+1}(\theta^{\star})\}}{\sum\limits_{p=1}^{m}\min\{\{\mu_{j,t'(\omega)}(\theta_p)\}_{{j\in\mathcal{N}_i[t'(\omega)]\cup\{i\}}},\pi_{i,t'(\omega)+1}(\theta_p)\}}\\
&\overset{(b)}{\geq}\frac{\eta(\omega)}{\sum\limits_{p=1}^{m}\min\{\{\mu_{j,t'(\omega)}(\theta_p)\}_{{j\in\mathcal{N}_i[t'(\omega)]\cup\{i\}}},\pi_{i,t'(\omega)+1}(\theta_p)\}}\\
&\overset{}{\geq}\frac{\eta(\omega)}{\sum\limits_{p=1}^{m}\pi_{i,t'(\omega)+1}(\theta_p)}
\overset{(c)}{=}\eta(\omega),
\end{aligned}
\label{eqn:lower1}
\end{equation}
where $(a)$ is given by \eqref{eqn:rule1}, $(b)$ follows from the way $\eta(\omega)$ is defined and by noting that $\pi_{i,t}(\theta^{\star}) \geq \eta(\omega), \forall t\geq t'(\omega), \forall i\in\mathcal{V}$, and $(c)$ follows by noting that the local belief vectors generated via \eqref{eqn:Bayes} are valid probability distributions over the hypothesis set $\Theta$ at each time-step, and hence $\sum\limits_{p=1}^{m}\pi_{i,t'(\omega)+1}(\theta_p)=1$. The above reasoning applies to every agent in the network, and can be repeated to establish \eqref{eqn:lowerbound} via induction.
\end{proof}

The next result establishes that the intrinsic discriminatory capabilities of an agent are preserved under our learning rule.
\begin{lemma}
Suppose the conditions stated in Theorem \ref{thm:main} hold, and Algorithm \ref{Algo:Algo1} is employed by each agent. Consider any false hypothesis $\theta\in\Theta\setminus\{\theta^{\star}\}$, and an
agent $i\in\mathcal{S}(\theta^{\star},\theta)$. Then, 
 \begin{equation}
        \liminf_{t\to\infty}-\frac{\log\mu_{i,t}(\theta)}{t} \geq K_i(\theta^{\star},\theta) \hspace{1mm} a.s.
        \label{eqn:source_rate}
    \end{equation}
    \label{lemma:source}
\end{lemma}
\begin{proof}
With $\bar{\Omega}$ defined as in Lemma \ref{lemma:lower_bound}, recall that $\mathbb{P}^{\theta^{\star}}(\bar{\Omega})=1$, and pick any $\omega\in\bar{\Omega}.$ Now consider any false hypothesis $\theta\in\Theta\setminus\{\theta^{\star}\}$, and an
agent $i\in\mathcal{S}(\theta^{\star},\theta)$.  Fix any $\epsilon > 0$, and notice that since $i\in\mathcal{S}(\theta^{\star},\theta)$, Eq. \eqref{eqn:localrate} in Lemma \ref{lemma:Bayes} implies that there exists ${t}_i(\omega,\theta,\epsilon)$, such that
\begin{equation}
    \pi_{i,t}(\theta) <  e^{-(K_i(\theta^{\star},\theta)-\epsilon)t}, \forall t \geq {t}_i(\omega,\theta,\epsilon).
    \label{eqn:source_bound1}
\end{equation}
Furthermore, since $\omega\in\bar{\Omega}$, Lemma \ref{lemma:lower_bound} guarantees the existence of a time-step $t'(\omega) \in (0,\infty)$, and a constant $\eta(\omega)\in(0,1)$, such that on $\omega$, $\pi_{i,t}(\theta^{\star}) \geq \eta(\omega), \mu_{i,t}(\theta^{\star}) \geq \eta(\omega), \forall t\geq t'(\omega), \forall i\in\mathcal{V}$. Let $\bar{t}_i(\omega,\theta,\epsilon)=\max\{t'(\omega),t_i(\omega,\theta,\epsilon)\}$. Let us suppress the dependence of $\bar{t}_i(\omega,\theta,\epsilon)$ on $i,\omega,\theta$ and $\epsilon$ for simplicity of notation, and observe the following inequalities:
\begin{equation}
\begin{aligned}
\mu_{i,\bar{t}+1}(\theta)&\overset{(a)}{\leq}\frac{\pi_{i,\bar{t}+1}(\theta)}{\sum\limits_{p=1}^{m}\min\{\{\mu_{j,\bar{t}}(\theta_p)\}_{{j\in\mathcal{N}_i[\bar{t}]\cup\{i\}}},\pi_{i,\bar{t}+1}(\theta_p)\}}\\
&\overset{}{\leq}\frac{\pi_{i,\bar{t}+1}(\theta)}{\min\{\{\mu_{j,\bar{t}}(\theta^{\star})\}_{{j\in\mathcal{N}_i[\bar{t}]\cup\{i\}}},\pi_{i,\bar{t}+1}(\theta^{\star})\}}\\
&\overset{(b)}{<}\frac{e^{-(K_i(\theta^{\star},\theta)-\epsilon)(\bar{t}+1)}}{\eta(\omega)}.
\end{aligned}
\label{eqn:src_upperbnd1}
\end{equation}
In the above inequalities, (a) follows from \eqref{eqn:rule1}, whereas (b) follows from \eqref{eqn:source_bound1} and by noting that all agents have both their local and actual beliefs lower bounded by $\eta(\omega)$ beyond time-step $\bar{t}.$ In particular, it is easy to see that the arguments used to arrive at \eqref{eqn:src_upperbnd1} apply to each time-step $t\geq\bar{t}+1.$ Based on \eqref{eqn:src_upperbnd1}, we then obtain that $\forall t\geq\bar{t}+1$: 
\begin{equation}
-\frac{\log\mu_{i,t}(\theta)}{t} > {(K_i(\theta^{\star},\theta)-\epsilon)}+\frac{\log \eta(\omega)}{t}.
\label{eqn:src_bnd2}
\end{equation}
Taking the limit inferior on both sides of \eqref{eqn:src_bnd2}, and noting that $\epsilon$ can be made arbitrarily small, readily leads to \eqref{eqn:source_rate}.
\end{proof}

For the subsequent discussion, let us fix a particular false hypothesis $\theta\in\Theta\setminus\{\theta^{\star}\}$, and assume that global identifiability holds. Let $v_{\theta}\in\argmax_{l\in\mathcal{S}(\theta^{\star},\theta)}K_l(\theta^{\star},\theta)$ represent any agent with the best discriminatory power w.r.t. the false hypothesis $\theta$, given that $\theta^{\star}$ gets realized. Based on Lemma \ref{lemma:source}, we have
\begin{equation}
        \liminf_{t\to\infty}-\frac{\log\mu_{v_{\theta},t}(\theta)}{t} \geq K_{v_{\theta}}(\theta^{\star},\theta) \hspace{1mm} a.s.
        \label{eqn:bst_src}
    \end{equation}
Our goal is to now establish that each agent $i\in\mathcal{V}\setminus\{v_{\theta}\}$ inherits the same asymptotic rate of rejection of $\theta$ as that of agent $v_{\theta}$ in \eqref{eqn:bst_src}. Roughly speaking, we will achieve this by showing that under the assumption of joint strong-connectivity, the belief of any agent $i\in\mathcal{V}\setminus\{v_{\theta}\}$ on $\theta$ is ``not too far off" from the belief of agent $v_{\theta}$ on $\theta$. In what follows, we make this idea precise. First, we require some additional notation: with each agent $i\in\mathcal{V}$, we associate a non-negative scalar $c_{i,t}(\theta)\in[0,\infty]$. These parameters evolve based on the following rules.\footnote{Note that the agents do not actually maintain or update the parameters $c_{i,t}(\theta)$. Instead, they have been introduced solely for the purpose of analysis.} 
\begin{enumerate}
    \item[(i)] $c_{v_{\theta},t}(\theta)=0, \forall t\in\mathbb{N}$.
    \item[(ii)] $c_{i,0}(\theta)=\infty, \forall i\in\mathcal{V}\setminus\{v_{\theta}\}$. 
    \item[(iii)] For each $i\in\mathcal{V}\setminus\{v_{\theta}\}$ and $t\in\mathbb{N}$, define $\tau_{i,t}(\theta)\triangleq\min_{j\in\mathcal{N}_i[t]\cup\{i\}}c_{j,t}(\theta)$, and 
    \begin{equation}
        c_{i,t+1}(\theta)\triangleq\tau_{i,t}(\theta)+1.
        \label{eqn:freshness_update}
    \end{equation}
\end{enumerate}
To explain the purpose of the above rules, we will adhere to the following terminology. We say that there exists a path of length $m\in\mathbb{N}_{+}$ from $v_{\theta}$ to $i\in\mathcal{V}\setminus\{{v}_{\theta}\}$ over $[t-m,t-1]$, if there exist agents $x(t-m+1),\ldots,x(t)\in\mathcal{V}\setminus\{{v}_{\theta}\}$, such that $(x(\tau-1),x(\tau))\in\mathcal{E}[\tau-1]$, where $\tau\in\{t-m+1,\ldots,t\}, x(t-m)=v_{\theta}$, and $x(t)=i$. Note that the agents appearing in the path need not be distinct, and that we have assumed the presence of self-loops in each graph $\mathcal{G}[t],t\in\mathbb{N}$. Rules (i)-(iii) have been designed in a manner such that if $c_{i,t}(\theta)$ is finite at any time-step $t\in\mathbb{N}$ for any agent $i\in\mathcal{V}\setminus\{{v}_{\theta}\}$, then there exists a path of length $c_{i,t}(\theta)$ from $v_{\theta}$ to $i$ over $[t-c_{i,t}(\theta),t-1]$, in the sense described above. Analyzing the time-evolution of $c_{i,t}(\theta)$ enables us to then relate the belief  $\mu_{i,t}(\theta)$ of agent $i$ to a delayed-version of the belief $\mu_{v_{\theta},t}(\theta)$ of agent $v_{\theta}$, where the delay is precisely $c_{i,t}(\theta)$ (the above statements are formalized and proven in Lemma \ref{lemma:belief_relation}). Since agent $v_{\theta}$ is the reference agent here, its delay w.r.t. its own belief on $\theta$ is set to $0$ for all time, thus explaining rule (i). Initially, all agents in $\mathcal{V}\setminus\{v_{\theta}\}$ start out with an ``infinite-delay
" w.r.t. the belief of agent $v_{\theta}$; this is captured by rule (ii). Finally, the rationale behind updating $c_{i,t}(\theta)$ via rule (iii) is to formalize the intuition that under the assumption of joint strong-connectivity, the lengths of paths linking $v_{\theta}$ to agents in $\mathcal{V}\setminus\{v_{\theta}\}$ (and hence, the corresponding delays) should eventually remain uniformly bounded; we begin by establishing this fact in the following lemma.

\begin{lemma}
Consider any $\theta\in \Theta\setminus\{\theta^{\star}\}$ and suppose the joint strong-connectivity assumption (Assumption \ref{assump:connectivity}) holds. Then, the following is true:
   \begin{equation}
        c_{i,t}(\theta) \leq 2(n-1)T, \forall i\in\mathcal{V}, \forall t \geq (n-1)T,
        \label{eqn:delay_bnd}
    \end{equation}
where $T$ is the constant appearing in Assumption \ref{assump:connectivity}.
\label{lemma:delay_bound}
\end{lemma}
\begin{proof}
Observe that the conclusion in \eqref{eqn:delay_bnd} is trivially true for agent $v_{\theta}$ since $c_{v_{\theta},t}(\theta)=0, \forall t\in\mathbb{N}$. To prove the result for agents in the set $\mathcal{V}\setminus\{v_{\theta}\}$, we begin by claiming that
\begin{equation}
    c_{i,(n-1)T}(\theta) \leq (n-1)T, \forall i\in\mathcal{V}.
        \label{eqn:delay_finite}
\end{equation}
To prove this claim, let $\mathcal{L}_0(\theta^{\star},\theta)=\{v_{\theta}\}$, and define
\begin{equation}
    \mathcal{L}_1(\theta^{\star},\theta)\triangleq\{i\in\mathcal{V}\setminus\mathcal{L}_0(\theta^{\star},\theta): \{\bigcup \limits_{\tau=0}^{T-1}\mathcal{N}_i[\tau]\}\cap\mathcal{L}_0(\theta^{\star},\theta)\neq\emptyset\}
\end{equation}
as the set of agents in $\mathcal{V}\setminus\{v_{\theta}\}$ that have a direct edge from agent $v_{\theta}$ at least once over the interval $[0,T)$. Assumption \ref{assump:connectivity} implies that $ \mathcal{L}_1(\theta^{\star},\theta)$ is non-empty (barring the trivial case when $\mathcal{V}=\{v_{\theta}\}$). Now pick any agent $i\in \mathcal{L}_1(\theta^{\star},\theta)$, and notice that since $v_{\theta}\in\mathcal{N}_{i}[\tau]$ for some $\tau\in[0,T)$, update rule \eqref{eqn:freshness_update} implies $c_{i,\tau+1}(\theta)=1$.\footnote{Notice that based on the update rule \eqref{eqn:freshness_update}, $c_{i,t}(\theta) \geq 1,  \forall i\in\mathcal{V}\setminus\{v_{\theta}\}$. Thus, $\argmin _{{j\in\mathcal{N}_i[t]\cup\{i\}}}c_{j,t}(\theta)=v_{\theta}$ whenever $v_{\theta}\in\mathcal{N}_{i}[t]$, since $c_{v_{\theta},t}(\theta)=0, \forall t\in\mathbb{N}.$} In particular, based on \eqref{eqn:freshness_update},
\begin{equation}
    c_{i,t+1}(\theta)\leq c_{i,t}(\theta)+1.
    \label{eqn:freshness_growth}
\end{equation}
Based on the above discussion, it follows that for each agent $i\in\mathcal{L}_1(\theta^{\star},\theta)$, $c_{i,T}(\theta)\leq T.$ The claim in \eqref{eqn:delay_finite} follows readily for each agent $i\in\mathcal{L}_1(\theta^{\star},\theta)$ by appealing to \eqref{eqn:freshness_growth}. 
Let us now recursively define the sets $\mathcal{L}_r(\theta^{\star},\theta),1\leq r \leq (n-1)$, as
\begin{equation}
    \mathcal{L}_r(\theta^{\star},\theta)\triangleq\{i\in\mathcal{V}\setminus\bigcup \limits_{q=0}^{(r-1)}\mathcal{L}_q(\theta^{\star},\theta):\{\hspace{-2.5mm}\bigcup \limits_{\tau=(r-1)T}^{rT-1}\hspace{-3mm}\mathcal{N}_i[\tau]\}\cap\{\bigcup \limits_{q=0}^{(r-1)}\mathcal{L}_q(\theta^{\star},\theta)\}\neq \emptyset\}.
\end{equation}
In words, $\mathcal{L}_r(\theta^{\star},\theta)$ are those agents belonging to $\mathcal{V}\setminus\bigcup \limits_{q=0}^{(r-1)}\mathcal{L}_q(\theta^{\star},\theta)$ that each have at least one neighbor from the set $\bigcup \limits_{q=0}^{(r-1)}\mathcal{L}_q(\theta^{\star},\theta)$ over the interval $[(r-1)T,rT-1]$. We complete the proof of the claim by inducting on $r$. The base case with $r=1$ has already been proven above. Now suppose the following is true: $c_{i,rT}(\theta)\leq rT,\forall i\in \mathcal{L}_r(\theta^{\star},\theta)$, where $r\in\{1,\ldots,m-1\}$, and $m\in\{2,\ldots,n-1\}.$ Let $r=m.$ If $\mathcal{V}\setminus\bigcup \limits_{q=0}^{(m-1)}\mathcal{L}_q(\theta^{\star},\theta)$ is empty, then we are done. Else, based on Assumption \ref{assump:connectivity}, it must be that $\mathcal{L}_{m}(\theta^{\star},\theta)$ is non-empty. Pick any agent $i\in\mathcal{L}_{m}(\theta^{\star},\theta)$, and notice that it has a neighbor $j$ (say) from the set $\bigcup\limits_{q=0}^{(m-1)}\mathcal{L}_q(\theta^{\star},\theta)$ at some time-step $\tau \in [(m-1)T,mT)$. The induction hypothesis coupled with \eqref{eqn:freshness_growth} implies that $c_{j,\tau}(\theta) \leq \tau$, and hence $c_{i,\tau+1}(\theta)\leq c_{j,\tau}(\theta)+1\leq \tau+1$ based on \eqref{eqn:freshness_update}. Appealing to \eqref{eqn:freshness_growth} then reveals that $c_{i,mT}(\theta) \leq mT$, thus completing the induction step. Finally, noting that $\bigcup\limits_{q=0}^{(n-1)}\mathcal{L}_q(\theta^{\star},\theta)=\mathcal{V}$ completes our proof of the claim \eqref{eqn:delay_finite}. An identical line of argument as above can be employed to show that $c_{i,2(n-1)T} \leq (n-1)T, \forall i\in\mathcal{V}$. In particular, this can be done by first taking $\mathcal{C}_0(\theta^{\star},\theta)=\{v_{\theta}\}$, and recursively defining the sets 
 $\mathcal{C}_r(\theta^{\star},\theta),1\leq r \leq (n-1)$ as
\begin{equation}
    \mathcal{C}_r(\theta^{\star},\theta)\triangleq\{i\in\mathcal{V}\setminus\bigcup \limits_{q=0}^{(r-1)}\mathcal{C}_q(\theta^{\star},\theta):\{\hspace{-3.5mm}\bigcup \limits_{\tau=(n+r-2)T}^{(n+r-1)T-1}\hspace{-4.5mm}\mathcal{N}_i[\tau]\}\cap\{\bigcup \limits_{q=0}^{(r-1)}\mathcal{C}_q(\theta^{\star},\theta)\}\neq \emptyset\}.
\end{equation}
One can then easily prove via induction that
$c_{i,(n-1+r)T}(\theta) \leq rT, \forall i\in \mathcal{C}_{r}(\theta^{\star},\theta)$, where $1\leq r \leq (n-1)$. The rest then follows from \eqref{eqn:freshness_growth}.

We can repeat the above argument to establish that $c_{i,m(n-1)T}(\theta)\leq(n-1)T, \forall i\in\mathcal{V}, \forall m\in\mathbb{N}_{+}$. Finally, based on the above bound and \eqref{eqn:freshness_growth}, it follows that for each agent $i\in\mathcal{V}$, $c_{i,t}(\theta)$ is upper-bounded by $2(n-1)T$ at any time-step $t\in(m(n-1)T,(m+1)(n-1)T)$, where $m\in\mathbb{N}_{+}$. This establishes \eqref{eqn:delay_bnd} and completes the proof.
\end{proof}

The next lemma relates $\mu_{i,t}(\theta), i\in\mathcal{V}\setminus\{v_{\theta}\}$ to $\mu_{v_{\theta},t}(\theta)$ in terms of the parameter $c_{i,t}(\theta)$ and, in turn, provides the final ingredient required to prove Theorem \ref{thm:main}.   

\begin{lemma}
Consider any $\theta\in \Theta\setminus\{\theta^{\star}\}$. Suppose the joint strong-connectivity assumption holds (Assumption \ref{assump:connectivity}), and each agent applies Algorithm \ref{Algo:Algo1}. Suppose $c_{i,t}(\theta)$ is finite, where $i\in\mathcal{V}\setminus\{v_{\theta}\}$, and $t\in\mathbb{N}$. Then, the following are true.
\begin{itemize}
    \item[(i)] There exists a path of length $c_{i,t}(\theta)$ from $v_{\theta}$ to $i$ over $[t-c_{i,t}(\theta),t-1]$.
    \item[(ii)] Let the path linking $v_{\theta}$ to $i$ over $[t-c_{i,t}(\theta),t-1]$ in part (i) be denoted $x(t-c_{i,t}(\theta)),x(t-c_{i,t}(\theta)+1),\ldots,x(t)$, where $x(t-c_{i,t}(\theta))=v_{\theta}$ and $x(t)=i$. Then
\begin{equation}
    \mu_{i,t}(\theta) \leq \frac{\mu_{v_{\theta},a_{i,t}(\theta)}(\theta)}{\prod\limits_{\tau=a_{i,t}(\theta)+1}^{t}\eta_{x(\tau),\tau}(\theta^{\star})}, 
    \label{eqn:belief_relation}
\end{equation}
where $a_{i,t}(\theta)=t-c_{i,t}(\theta)$, and
\begin{equation}
    \eta_{i,t}(\theta^{\star})\triangleq\min\{\{\mu_{j,{t-1}}(\theta^{\star})\}_{{j\in\mathcal{N}_i[{t-1}]\cup\{i\}}},\pi_{i,{t}}(\theta^{\star})\}, \forall i\in\mathcal{V}.
    \label{eqn:eta_defn}
\end{equation}
\end{itemize}
\label{lemma:belief_relation}
\end{lemma}
\begin{proof}
We prove part (i) by inducting on the value of $c_{i,t}(\theta)$. For the base case, suppose $c_{i,t}(\theta)=1$ for some agent $i\in\mathcal{V}\setminus\{v_{\theta}\}$ at some time-step $t$. Based on \eqref{eqn:freshness_update}, notice that this can happen if and only if $v_{\theta}\in\mathcal{N}_i[t-1]$; the claim in part (i) then follows readily for the base case. Fix an integer $m\geq2$, and suppose that the assertion of part (i) holds for any agent $i\in\mathcal{V}\setminus\{v_{\theta}\}$ and at any time-step $t$, whenever $c_{i,t}(\theta)\in\{1,\dots,m-1\}$. Now suppose that at some time-step $t$, $c_{i,t}(\theta)=m$ for some agent $i\in\mathcal{V}\setminus\{v_{\theta}\}.$ Referring to \eqref{eqn:freshness_update}, this is true only if $c_{l,t-1}(\theta)=m-1$ for some $l\in\mathcal{N}_{i}[t-1]\cup\{i\}.$ Since $m\geq2$, we have $c_{l,t-1}(\theta)\geq 1$, and hence $l\in\mathcal{V}\setminus\{v_{\theta}\}$. The induction hypothesis thus applies to agent $l$, implying the existence of a path of length $m-1$ from $v_{\theta}$ to $l$ over $[(t-1)-c_{l,t-1}(\theta),t-2]$, i.e., over $[t-m,t-2]$. Appending this path with the edge $(l,i)\in\mathcal{E}[t-1]$ immediately leads to the desired conclusion. 

For part (ii), consider the path $x(t-c_{i,t}(\theta)),x(t-c_{i,t}(\theta)+1),\ldots,x(t)$ from $v_{\theta}$ to $i$ over $[t-c_{i,t}(\theta),t-1]$, where $x(t-c_{i,t}(\theta))=v_{\theta}$ and $x(t)=i$. By definition of this path, $x(\tau-1)\in\mathcal{N}_{x(\tau)}[\tau-1]\cup\{x(\tau)\}$, for all $\tau\in\{a_{i,t}(\theta)+1,\ldots,t\}$. Thus, referring to \eqref{eqn:rule1}, we obtain
\begin{equation}
\begin{aligned}
\mu_{x(\tau),{\tau}}(\theta)&\overset{}{\leq}\frac{\mu_{x(\tau-1),\tau-1}(\theta)}{\sum\limits_{p=1}^{m}\min\{\{\mu_{j,{\tau-1}}(\theta_p)\}_{{j\in\mathcal{N}_{x(\tau)}[{\tau-1}]\cup\{x(\tau)\}}},\pi_{x(\tau),{\tau}}(\theta_p)\}}\\
&\overset{}{\leq}\frac{\mu_{x(\tau-1),\tau-1}(\theta)}{\eta_{x(\tau),\tau}(\theta^{\star})}.
\end{aligned}
\end{equation}
Using the above inequality recursively with $\tau\in\{a_{i,t}(\theta)+1,\ldots,t\}$ immediately leads to \eqref{eqn:belief_relation}. 
\end{proof}
\begin{proof}
\textbf{(Theorem 1)}: Fix a false hypothesis $\theta\in\Theta\setminus\{\theta^{\star}\}$. Based on the assumption of global identifiability, note that the set $\mathcal{S}(\theta^{\star},\theta)$ is non-empty. Recall that $v_{\theta}$ is any agent for which $K_i(\theta^{\star},\theta),i\in\mathcal{S}(\theta^{\star},\theta)$ is maximum, and note that we have already established that the assertion of Theorem \ref{thm:main}, namely inequality \eqref{eqn:asymprate}, holds for agent $v_{\theta}$ in Lemma \ref{lemma:source}. Now consider an agent $i\in\mathcal{V}\setminus\{v_{\theta}\}$, and notice that if $t\geq(n-1)T$, then $c_{i,t}(\theta)$ is uniformly bounded based on Lemma \ref{lemma:delay_bound}. Thus, the assertions in Lemma \ref{lemma:belief_relation} hold for all $t\geq(n-1)T$. Taking the natural log on both sides of \eqref{eqn:belief_relation}, dividing throughout by $t$, and simplifying, we obtain the following for all $t\geq (n-1)T$:
\begin{equation}
    -\frac{\log\mu_{i,t}(\theta)}{t}\geq -\frac{\log\mu_{v_{\theta},a_{i,t}(\theta)}(\theta)}{t}+\hspace{-4mm}\sum\limits_{\tau=a_{i,t}(\theta)+1}^{t}\hspace{-3mm}\frac{\log\eta_{x(\tau),\tau}(\theta^{\star})}{t},
    \label{eqn:ineq1}
\end{equation}
where $a_{i,t}(\theta)=t-c_{i,t}(\theta)$, $\eta_{i,t}(\theta^{\star})$ is as defined in \eqref{eqn:eta_defn}, and $x(\tau),\tau\in\{a_{i,t}(\theta)+1,\ldots,t\}$, are agents in the path linking $v_{\theta}$ to $i$ over $[a_{i,t}(\theta),t-1]$. For the remainder of the proof, to lighten the notation, let us drop the subscript on $v_{\theta}$, and let $a(t)=a_{i,t}(\theta)$. Based on \eqref{eqn:rule1}, we then have:
\begin{equation}
    \mu_{v,a(t)}(\theta) \leq \frac{\pi_{v,a(t)}(\theta)}{\eta_{v,a(t)}(\theta^{\star})}. 
\end{equation}
A bit of straightforward algebra then yields:
\begin{equation}
    -\frac{\log\mu_{v,a(t)}(\theta)}{t} \geq -\frac{\log \pi_{v,t}(\theta)}{t}+\frac{\log\frac{\pi_{v,t}(\theta)}{\pi_{v,a(t)}(\theta)}}{t}+\frac{\log\eta_{v,a(t)}(\theta^{\star})}{t}.
    \label{eqn:ineq2}
\end{equation}
Combining \eqref{eqn:ineq1} and \eqref{eqn:ineq2}, we obtain for $t\geq(n-1)T$:
\begin{equation}
   -\frac{\log\mu_{i,t}(\theta)}{t}\geq -\frac{\log \pi_{v,t}(\theta)}{t}+b(t),
\label{eqn:ineq_main}
\end{equation}
where $b(t)=b_1(t)+b_2(t)+b_3(t)$,
\begin{equation}
    b_1(t)=\hspace{-3mm}\sum\limits_{\tau=a(t)+1}^{t}\hspace{-3mm}\frac{\log\eta_{x(\tau),\tau}(\theta^{\star})}{t}, \hspace{1mm} b_2(t)=\frac{\log\frac{\pi_{v,t}(\theta)}{\pi_{v,a(t)}(\theta)}}{t},
    \label{eqn:terms1}
\end{equation}
and 
\begin{equation}
    b_3(t)=\frac{\log\eta_{v,a(t)}(\theta^{\star})}{t}.
    \label{eqn:terms2}
\end{equation}
We now argue that each of the terms $b_1(t),b_2(t)$ and $b_3(t)$ converge to 0 almost surely as $t\to\infty$. To do so, recall that the set $\bar{\Omega}\subseteq\Omega$ in Lemma \ref{lemma:lower_bound} has measure 1. In what follows, we prove that $b_1(t),b_2(t)$ and $b_3(t)$ converge to 0 for each sample path $\omega\in\bar{\Omega}.$ Accordingly, fix $\omega\in\bar{\Omega}$, and recall $\eta(\omega)\in(0,1)$ and $t'(\omega)\in(0,\infty)$ from Lemma \ref{lemma:lower_bound}. Suppose $t > t'(\omega)+2\bar{T}$, where $\bar{T}=(n-1)T$. We then claim the following: 
\begin{equation}
    \pi_{l,\tau}(\theta^{\star}) \geq \eta(\omega), \mu_{l,\tau}(\theta^{\star}) \geq \eta(\omega), \forall l\in\mathcal{V}, \forall \tau \geq a(t).
    \label{eqn:lowerbnd2}
\end{equation}
To see why this is true, notice that based on Lemma \ref{lemma:delay_bound}, the following holds when $t > t'(\omega)+2\bar{T}$: 
\begin{equation}
a(t)=t-c_{i,t}(\theta)\geq t-2\bar{T}>t'(\omega).
\label{eqn:time_ineq}
\end{equation}
The claim regarding \eqref{eqn:lowerbnd2} then follows readily from equation \eqref{eqn:lowerbound} in Lemma \ref{lemma:lower_bound}. Based on the above discussion, and referring to \eqref{eqn:eta_defn}, we immediately note that when $t > t'(\omega)+2\bar{T}$, 
\begin{equation}
    \eta_{l,\tau}(\theta^{\star}) \geq \eta(\omega), \forall l\in\mathcal{V}, \forall \tau \geq a(t).
    \label{eqn:ineq_eta}
\end{equation}
For establishing the convergence of $b_1(t),b_2(t)$ and $b_3(t)$, suppose $t > t'(\omega)+2\bar{T}$. Regarding $b_1(t)$, we then observe:
\begin{equation}
    \begin{aligned}
    |b_1(t)|&\overset{}{=}\left|\sum\limits_{\tau=a(t)+1}^{t}\hspace{-3mm}\frac{\log\eta_{x(\tau),\tau}(\theta^{\star})}{t}\right|\\
    &\overset{(a)}{\leq}\sum\limits_{\tau=a(t)+1}^{t}\hspace{-3mm}\frac{\left|\log\eta_{x(\tau),\tau}(\theta^{\star})\right|}{t}\\
    &\overset{(b)}{\leq}\frac{(t-a(t))}{t}\log\frac{1}{\eta(\omega)}\\
    &\overset{(c)}{\leq}\frac{2\bar{T}}{t}\log\frac{1}{\eta(\omega)},
    \end{aligned}
    \label{eqn:b_1}
\end{equation}
where (a) follows from the triangle inequality, (b) follows from \eqref{eqn:ineq_eta}, and (c) follows from \eqref{eqn:time_ineq}. From \eqref{eqn:b_1}, we immediately note that $b_1(t)\to 0$ along $\omega$. Let us now turn our attention to $b_2(t)$, and take note of the following:
\begin{equation}
 \begin{aligned}
    |b_2(t)|&\overset{(a)}{=}\frac{1}{t}\left|\log\frac{\pi_{v,t}(\theta^{\star})}{\pi_{v,a(t)}(\theta^{\star})}+\hspace{-2.5mm}\sum\limits_{\tau=a(t)+1}^{t}\hspace{-2.5mm}\log\frac{l_v(s_{v,\tau}|\theta)}{l_v(s_{v,\tau}|\theta^{\star})}\right|\\
    &\overset{(b)}{\leq}\frac{1}{t}\left|\log\frac{\pi_{v,t}(\theta^{\star})}{\pi_{v,a(t)}(\theta^{\star})}\right|+\frac{1}{t}\hspace{-2mm}\sum\limits_{\tau=a(t)+1}^{t}\left| \log\frac{l_v(s_{v,\tau}|\theta)}{l_v(s_{v,\tau}|\theta^{\star})}\right|\\
    &\overset{(c)}{\leq}\frac{2}{t}\log\frac{1}{\eta(\omega)}+\frac{(t-a(t))L}{t}\\
    &\overset{(d)}{\leq}\frac{2}{t}\left(\log\frac{1}{\eta(\omega)}+L\bar{T}\right),
    \end{aligned}
    \label{eqn:b_2}
\end{equation}
where (a) follows from \eqref{eqn:recursion} and some simple manipulations, (b) is a consequence of the triangle inequality, (c) follows from \eqref{eqn:bounded_lograt} and \eqref{eqn:lowerbnd2}, and (d) follows from \eqref{eqn:time_ineq}. Based on \eqref{eqn:b_2}, we then note that $b_2(t)\to 0$ along $\omega$. Finally, the fact that $b_3(t)$ converges to 0 along $\omega$ follows immediately by appealing to \eqref{eqn:ineq_eta}. We have thus established that $b(t)\to 0$ almost surely. The desired conclusion then follows by taking the limit inferior on both sides of \eqref{eqn:ineq_main}, and noting that
\begin{equation}
  \lim_{t\to\infty} -\frac{\log \pi_{v,t}(\theta)}{t}=  \lim_{t\to\infty}-\frac{1}{t}\rho_{v,t}(\theta)=K_v(\theta^{\star},\theta)\hspace{1mm} a.s.,
\end{equation}
where $\rho_{v,t}(\theta)$ is as defined in Lemma \ref{lemma:Bayes}. The fact that $\mu_{i,t}(\theta)\to 0$ is immediate, since $K_v(\theta^{\star},\theta)>0$ based on global identifiability. The above analysis applies identically to each $\theta\in\Theta\setminus\{\theta^{\star}\}$. This establishes consistency of our rule, and completes the proof. 
\end{proof}

\section{Proof of Theorem
\ref{thm:conc}}
\label{sec:App2}
\noindent To prove Theorem \ref{thm:conc}, we will make use of one of Littlewood's three principles: every pointwise convergent sequence of measurable functions is nearly uniformly convergent.
\begin{theorem} (\textbf{Egoroff's Theorem}) \cite[Chapter 18]{royden} Let $(X,\mathcal{M},\mu)$ be a finite measure space and $\{f_n\}$ a sequence of measurable functions on $X$ that converge pointwise a.e. (almost everywhere) on $X$ to a function $f$ that is finite a.e. on $X$. Then for each $\epsilon >0$, there is a measurable subset $X_{\epsilon}$ of $X$
for which ${f_n}\rightarrow f$ uniformly on $X_{\epsilon}$, and $\mu(X_{\epsilon}) \geq 1-\epsilon.$
\end{theorem}
\begin{proof} (\textbf{Theorem 2}): Consider a $\theta\in\Theta\setminus\{\theta^{\star}\}$, and recall that $K_{v_{\theta}}(\theta^{\star},\theta)=\max_{l\in\mathcal{S}(\theta^{\star},\theta)}K_l(\theta^{\star},\theta)=\bar{K}(\theta^{\star},\theta).$ We only prove the result for $i\in\mathcal{V}\setminus\{v_{\theta}\},$ since the argument for agent $v_{\theta}$ will be similar. To this end, let us fix an agent $i\in\mathcal{V}\setminus\{v_{\theta}\}$. We adhere to the notation used in the proof of Lemma \ref{lemma:Bayes}, and for simplicity assume that the initial local belief vectors $\boldsymbol{\pi}_{i,0}, i\in\mathcal{V}$ are uniform distributions over the hypothesis set $\Theta$; our subsequent arguments will continue to hold (with simple modifications) under the more general assumption on priors in line 1 of Algo \ref{Algo:Algo1}. We immediately note that based on the assumption of uniform priors, $\rho_{i,0}(\theta)=0, \forall i\in\mathcal{V}$. Now referring to inequality \eqref{eqn:ineq_main} in the proof of Theorem \ref{thm:main}, we obtain the following for $t\geq(n-1)T$:
\begin{equation}
    \begin{aligned}
    &\mathbb{P}^{\theta^{\star}}\left(-\frac{\log\mu_{i,t}(\theta)}{t} \leq \bar{K}(\theta^{\star},\theta)-\frac{\epsilon}{2}+b(t)\right)\\
    &\overset{(a)}{\leq}\mathbb{P}^{\theta^{\star}}\left(-\frac{\log\pi_{v_{\theta},t}(\theta)}{t} \leq \bar{K}(\theta^{\star},\theta)-\frac{\epsilon}{2}\right)\\
   &\overset{(b)}{\leq}\mathbb{P}^{\theta^{\star}}\left(-\frac{\rho_{v_{\theta},t}(\theta)}{t} \leq \bar{K}(\theta^{\star},\theta)-\frac{\epsilon}{2}\right)\\
   &\overset{(c)}{=}\mathbb{P}^{\theta^{\star}}\left(\frac{1}{t}\sum\limits_{k=1}^{t}\lambda_{v_{\theta},k}(\theta)-(-{K}_{v_{\theta}}(\theta^{\star},\theta)) \geq \frac{\epsilon}{2}\right)\\
   &\overset{(d)}{\leq}\exp(-\frac{\epsilon^2 t}{8L^2}).
   \label{eqn:conc_ineq1}
    \end{aligned}
\end{equation}
In the above steps, (a) follows directly from \eqref{eqn:ineq_main}, and (b) follows by noting that based on the definition of $\rho_{v_{\theta},t}(\theta)$,
\begin{equation}
    \frac{\log\pi_{v_{\theta},t}(\theta)}{t} \leq \frac{\rho_{v_{\theta},t}(\theta)}{t}, \forall t\in\mathbb{N}.
\end{equation}
Step (c) follows directly from \eqref{eqn:recursion} with $\rho_{v_{\theta},0}(\theta)=0.$ Finally, noting that $\frac{1}{t}\sum\limits_{k=1}^{t}\lambda_{v_{\theta},k}(\theta)\to -{K}_{v_{\theta}}(\theta^{\star},\theta)$ a.s. (as argued in the proof of Lemma \ref{lemma:Bayes}), using the fact that $|\lambda_{v_{\theta},t}(\theta)|\leq L, \forall t\in\mathbb{N}_{+}$ based on \eqref{eqn:bounded_lograt}, and applying Hoeffding's inequality \cite[Theorem 2]{hoeffding}, leads to (d). Now recall from the proof of Theorem \ref{thm:main} that $b(t)\to 0$ almost surely. Appealing to Egoroff's theorem, we then infer that given any arbitrarily small $\delta\in (0,1)$, there exists a set $\Omega'(\delta)\subseteq\Omega$ of $\mathbb{P}^{\theta^{\star}}$-measure at least $(1-\delta)$, such that $b(t)$ converges to $0$ uniformly on $\Omega'(\delta)$. Thus, given any $\epsilon >0$, there exists a $\omega$-independent constant $t(\epsilon,\delta)\in(0,\infty)$, such that $|b(t)|\leq \frac{\epsilon}{2},\forall t\geq t(\epsilon,\delta)$, along each sample path $\omega\in\Omega'(\delta)$. Setting $t'(\epsilon,\delta,n,T)=\max\{t(\epsilon,\delta),(n-1)T\}$, and referring to \eqref{eqn:conc_ineq1}, we immediately obtain that $\forall t \geq t'(\epsilon,\delta,n,T)$,
\begin{equation}
    \begin{aligned}
    &\mathbb{P}^{\theta^{\star}}\left(\left\{-\frac{\log\mu_{i,t}(\theta)}{t} \leq \bar{K}(\theta^{\star},\theta)-\epsilon\right\}\cap\Omega'(\delta)\right)\\
    &\leq \mathbb{P}^{\theta^{\star}}\left(\left\{-\frac{\log\mu_{i,t}(\theta)}{t} \leq \bar{K}(\theta^{\star},\theta)-\frac{\epsilon}{2}+b(t)\right\}\cap\Omega'(\delta)\right)\\
    &\leq \mathbb{P}^{\theta^{\star}}\left(-\frac{\log\mu_{i,t}(\theta)}{t} \leq \bar{K}(\theta^{\star},\theta)-\frac{\epsilon}{2}+b(t)\right) \leq \exp(-\frac{\epsilon^2 t}{8L^2}).
    \end{aligned}
\end{equation}
Taking the natural log on both sides of the resulting inequality, dividing throughout by $t$, simplifying, and then taking the limit inferior on both sides, leads to the desired result.
\end{proof}

\section{Proof of Theorem
\ref{thm:Byz}}
\label{sec:App3}
\begin{proof}
Consider an $f$-local adversarial set $\mathcal{A}\subset\mathcal{V}$, and let $\mathcal{R}=\mathcal{V}\setminus\mathcal{A}$. We study two separate cases. 

\underline{\textbf{Case 1:}} Consider a regular agent $i\in\mathcal{R}$ such that $|\mathcal{N}_i| < (2f+1)$. Based on the hypothesis of the theorem, we claim that $i\in\mathcal{S}(\theta_p,\theta_q)$, for every pair $\theta_p,\theta_q \in \Theta$. We prove this claim via contradiction. To do so, suppose there exists a pair $\theta_p,\theta_q\in\Theta$, such that $i\in\mathcal{V}\setminus\mathcal{S}(\theta_p,\theta_q)$. As $|\mathcal{N}_i| < (2f+1)$, the set $\{i\}$ is clearly not $(2f+1)$-reachable (see Def. \ref{defn:rreachable}). Thus, $\mathcal{G}$ is not strongly $(2f+1)$-robust w.r.t. the source set $\mathcal{S}(\theta_p,\theta_q)$, a fact that contradicts the hypothesis of the theorem. Thus, we have established that if the graph-theoretic condition identified in the theorem is met, then regular agents with fewer than $(2f+1)$ neighbors can distinguish between every pair of hypotheses. For such agents, the assertion of the theorem then follows directly from Lemma \ref{lemma:Bayes}, and update rules \eqref{eqn:Bayes} and \eqref{eqn:rule3}.

\underline{\textbf{Case 2:}} We now focus only on regular agents $i$ satisfying $|\mathcal{N}_i| \geq (2f+1)$. A key property of the LFRHE algorithm (Algo. \ref{Algo:Algo2}) that will be used throughout the proof is as follows. For any $i\in\mathcal{R}$, and any $\theta\in\Theta$, the filtering operation in line 7 of Algo. \ref{Algo:Algo2} ensures that at each $t\in\mathbb{N}$, we have
\begin{equation}
\mu_{j,t}(\theta) \in Conv(\Psi^{\theta}_{i,t}), \forall j \in \mathcal{M}^{\theta}_{i,t},
\label{eqn:LHRHE_property}
\end{equation}
where 
\begin{equation}
\Psi^{\theta}_{i,t} \triangleq \{\mu_{l,t}(\theta) \hspace{1mm} {:} \hspace{1mm}  l\in\mathcal{N}_i\cap\mathcal{R}\},
\label{eqn:set}
\end{equation}
and $Conv(\Psi^{\theta}_{i,t})$ is used to denote the convex hull formed by the points in the set $\Psi^{\theta}_{i,t}$ (recall that $\mathcal{M}^{\theta}_{i,t}$ was defined in line 8 of Algo \ref{Algo:Algo2} to be the set of agents in $\mathcal{N}_i$ whose beliefs are retained by agent $i$ after it removes the highest $f$ and lowest $f$ beliefs $\mu_{j,t}(\theta), j\in\mathcal{N}_i$). In words, any neighboring belief (on a particular hypothesis) that agent $i$ uses in the update rule \eqref{eqn:rule2} lies in the convex hull of the actual beliefs of its regular neighbors (on that particular hypothesis). To see why \eqref{eqn:LHRHE_property} is true, partition the neighbor set $\mathcal{N}_i$ of a regular agent into three sets $\mathcal{U}^{\theta}_{i,t}, \mathcal{M}^{\theta}_{i,t}$, and $\mathcal{J}^{\theta}_{i,t}$ as follows. Sets $\mathcal{U}^{\theta}_{i,t}$ and $\mathcal{J}^{\theta}_{i,t}$ are each of cardinality $f$, and contain neighbors of agent $i$ that transmit the highest $f$ and the lowest $f$ actual beliefs respectively, on the hypothesis $\theta$, to agent $i$ at time-step $t$. The set $\mathcal{M}^{\theta}_{i,t}$ contains the remaining neighbors of agent $i$, and is non-empty at every time-step since $|\mathcal{N}_i| \geq (2f+1)$. If $\mathcal{M}^{\theta}_{i,t}\cap\mathcal{A}=\emptyset$, then \eqref{eqn:LHRHE_property} holds trivially. Thus, consider the case when there are adversaries in the set $\mathcal{M}^{\theta}_{i,t}$, i.e., $\mathcal{M}^{\theta}_{i,t}\cap\mathcal{A} \neq \emptyset$. Given the $f$-locality of the adversarial model, and the nature of the filtering operation in the LFRHE algorithm, we infer that for each $j\in\mathcal{M}^{\theta}_{i,t}\cap\mathcal{A}$, there exist regular agents $u,v\in\mathcal{N}_i\cap\mathcal{R}$, such that $u\in\mathcal{U}^{\theta}_{i,t}$, $v\in\mathcal{J}^{\theta}_{i,t}$, and $\mu_{v,t}(\theta) \leq \mu_{j,t}(\theta) \leq \mu_{u,t}(\theta)$. This establishes our claim regarding equation \eqref{eqn:LHRHE_property}.

With the above property in hand, let $\bar{\Omega}\subseteq\Omega$ denote the set of sample paths for which assertions (i)-(iii) in Lemma \ref{lemma:Bayes} (Appendix \ref{sec:App1}) hold when restricted to the set of regular agents $\mathcal{R}$. Since the evolution of the local beliefs are unaffected by the presence of adversaries, Lemma \ref{lemma:Bayes} implies $\mathbb{P}^{\theta^{\star}}(\bar{\Omega})=1$. Now as in Lemma \ref{lemma:lower_bound}, fix a sample path $\omega\in\bar{\Omega}$. Define $\gamma_1 \triangleq \min_{i\in\mathcal{R}} \pi_{i,0}(\theta^{\star})$, pick a small number $\delta > 0$ satisfying $\delta < \gamma_1$, and observe that arguments similar to those in the proof of Lemma \ref{lemma:lower_bound} imply the existence of a time-step $t'(\omega)$, 
such that for all $t\geq t'(\omega), \pi_{i,t}(\theta^{\star}) \geq \gamma_1-\delta>0, \forall i \in \mathcal{R}.$ Let $\gamma_2(\omega)\triangleq\min_{i\in\mathcal{R}}\{\mu_{i,t'(\omega)}(\theta^{\star})\}$. As before, we claim $\gamma_2(\omega) > 0$. To establish this claim, we need to answer the following question: can an adversarial agent cause its out-neighbors to set their actual beliefs on $\theta^{\star}$ to be $0$ by setting its own actual belief on $\theta^{\star}$ to be $0$? We argue that this is impossible under the LFRHE algorithm. By way of contradiction, suppose there exists a time-step $\bar{t}(\omega)$ satisfying:
\begin{equation}
\bar{t}(\omega)=\min\{t\in\mathbb{N}\hspace{1mm}{:} \hspace{1mm} \exists i \in \mathcal{R} \hspace{1mm}\textrm{with} \hspace{1mm}\mu_{i,t}(\theta^{\star})=0\}.\
\end{equation}
In words, $\bar{t}(\omega)$ represents the first time-step when some regular agent $i$ sets its actual belief on the true hypothesis to be zero. Clearly, $\bar{t}(\omega)\neq 0$ based on line 1 of Algo. \ref{Algo:Algo2}. Suppose  $\bar{t}(\omega)$ is some positive integer, and focus on how agent $i$ updates $\mu_{i,\bar{t}(\omega)}(\theta^{\star})$ based on \eqref{eqn:rule2}. Following similar arguments as in the proof of Lemma \ref{lemma:lower_bound}, we know that $\pi_{i,t}(\theta^{\star}) > 0, \forall t\in \mathbb{N}, \forall i \in \mathcal{R}.$ At the same time, every belief featuring in the set $\Psi^{\theta^{\star}}_{i,\bar{t}(\omega)-1}$ (as defined in equation \eqref{eqn:set}) is strictly positive based on the way $\bar{t}(\omega)$ is defined. In light of the above arguments, and based on \eqref{eqn:LHRHE_property}, \eqref{eqn:set}, we infer:
\begin{equation}
\min\{\{\mu_{j,\bar{t}(\omega)-1}(\theta^{\star})\}_{j\in\mathcal{M}^{\theta^{\star}}_{i,\bar{t}(\omega)-1}},\pi_{i,\bar{t}(\omega)}(\theta^{\star})\} > 0.
\end{equation}
Thus, based on \eqref{eqn:rule2}, we must have $\mu_{i,\bar{t}(\omega)}(\theta^{\star}) > 0$, yielding the desired contradiction. With $\eta(\omega)\triangleq\min\{\gamma_1-\delta,\gamma_2(\omega)\} > 0$, one can easily  verify the following by referring to \eqref{eqn:rule2}:
\begin{equation}
\mu_{i,t}(\theta^{\star}) \geq \eta(\omega), \forall t \geq t'(\omega), \forall i\in\mathcal{R}.
\label{eqn:lower2}
\end{equation}
In particular, \eqref{eqn:lower2} follows by (i) noting that for each $i \in \mathcal{R}$, $\pi_{i,t'(\omega)+1}(\theta^{\star}) \geq \eta(\omega)$, and each belief featuring in the set $\Psi^{\theta^{\star}}_{i,t'(\omega)}$ is lower bounded by $\eta(\omega)$, (ii) leveraging \eqref{eqn:LHRHE_property}, \eqref{eqn:set}, and (iii) using a similar string of arguments as those used to arrive at \eqref{eqn:lower1}. Thus, we have established an analogous result as in Lemma \ref{lemma:lower_bound} for the regular agents. 

To proceed, let us fix a false hypothesis $\theta \neq \theta^{\star}$, and define $\tilde{K}(\theta^{\star},\theta)\triangleq \min_{v\in\mathcal{S}(\theta^{\star},\theta)\cap\mathcal{R}}K_v(\theta^{\star},\theta)$. Then, given any $\epsilon >0$, Lemma \ref{lemma:Bayes} implies the existence of a time-step $\tilde{t}_1(\omega,\theta,\epsilon)$, such that:
\begin{equation}
\pi_{i,t}(\theta) <  e^{-(\tilde{K}(\theta^{\star},\theta)-\epsilon)t}, \forall t \geq \tilde{t}_1(\omega,\theta,\epsilon), \forall i\in\mathcal{S}(\theta^{\star},\theta)\cap\mathcal{R}.
\label{eqn:src_bnd_Byz}
\end{equation}
Let $\tilde{t}_2=\max\{t'(\omega),\tilde{t}_1(\omega,\theta,\epsilon)\}$,
where we have suppressed the dependence of $\tilde{t}_2$ on $\omega,\theta$ and $\epsilon$. For any agent $i\in\mathcal{S}(\theta^{\star},\theta)\cap\mathcal{R}$, observe that based on \eqref{eqn:LHRHE_property}, \eqref{eqn:set} and \eqref{eqn:lower2},
\begin{equation}
\min\{\{\mu_{j,t}(\theta^{\star})\}_{j\in\mathcal{M}^{\theta^{\star}}_{i,t}},\pi_{i,t+1}(\theta^{\star})\} \geq \eta(\omega), \forall t \geq \tilde{t}_2. 
\end{equation}
Combining the above with a similar line of argument as used to arrive at \eqref{eqn:src_upperbnd1}, we obtain:
\begin{equation}
\mu_{i,t}(\theta) < C_1(\omega)e^{-(\tilde{K}(\theta^{\star},\theta)-\epsilon)t}, \forall t \geq \tilde{t}_2+1, \forall i \in \mathcal{S}(\theta^{\star},\theta)\cap\mathcal{R},
\label{eqn:Byz_src_uppbnd2}
\end{equation}
where $C_1(\omega)={\eta(\omega)}^{-1}.$
If $\mathcal{V}\setminus\mathcal{S}(\theta^{\star},\theta)$ is empty, then we are essentially done. Else, define
\begin{equation}
\mathcal{L}_1{(\theta^{\star},\theta)}\triangleq \{i\in \mathcal{V}\setminus\mathcal{S}(\theta^{\star},\theta) \hspace{1mm} {:} \hspace{1mm} |\mathcal{N}_i\cap\mathcal{S}(\theta^{\star},\theta)| \geq (2f+1)\}.
\label{eqn:level1thm2}
\end{equation}
Whenever $\mathcal{V}\setminus\mathcal{S}(\theta^{\star},\theta)$ is non-empty, we claim that $\mathcal{L}_1{(\theta^{\star},\theta)}$ (as defined above) is also non-empty based on the hypothesis of the theorem. To see this, note that if $\mathcal{L}_1{(\theta^{\star},\theta)}$ is empty, then $\mathcal{C}=\mathcal{V}\setminus\mathcal{S}(\theta^{\star},\theta)$ is not $(2f+1)$-reachable, violating the fact that $\mathcal{G}$ is strongly $(2f+1)$-robust w.r.t. $\mathcal{S}(\theta^{\star},\theta)$. We claim
that the following holds for each $i\in\mathcal{L}_1{(\theta^{\star},\theta)}\cap\mathcal{R}$:
\begin{equation}
\min_{j\in\mathcal{M}^{\theta}_{i,t}}\mu_{j,t}(\theta)< C_1(\omega)e^{-(\tilde{K}(\theta^{\star},\theta)-\epsilon)t}, \forall t \geq \tilde{t}_2+1.
\label{eqn:Byz_bound1}
\end{equation}
To verify the above claim, pick any agent $i\in\mathcal{L}_1{(\theta^{\star},\theta)}\cap\mathcal{R}$, and suppose $t\geq \tilde{t}_2+1$. When $|\mathcal{M}^{\theta}_{i,t}\cap\{\mathcal{S}(\theta^{\star},\theta)\cap\mathcal{R}\}|>0$, the claim follows immediately based on \eqref{eqn:Byz_src_uppbnd2}. Consider the case when $|\mathcal{M}^{\theta}_{i,t}\cap\{\mathcal{S}(\theta^{\star},\theta)\cap\mathcal{R}\}|=0$. Since $i\in\mathcal{L}_1{(\theta^{\star},\theta)}$, it has at least $(2f+1)$ neighbors in $\mathcal{S}(\theta^{\star},\theta)$, out of which at least $f+1$ are regular based on the $f$-locality of the adversarial model. Since the set $\mathcal{J}^{\theta}_{i,t}$ has cardinality $f$, it must then be that $|\mathcal{U}^{\theta}_{i,t}\cap\{\mathcal{S}(\theta^{\star},\theta)\cap\mathcal{R}\}| > 0$. Let $u\in\mathcal{U}^{\theta}_{i,t}\cap\{\mathcal{S}(\theta^{\star},\theta)\cap\mathcal{R}\}$. Based on the way $\mathcal{M}^{\theta}_{i,t}$ is defined, it must  be that $\mu_{j,t}(\theta) \leq \mu_{u,t}(\theta) < C_1(\omega)e^{-(\tilde{K}(\theta^{\star},\theta)-\epsilon)t}, \forall j \in \mathcal{M}^{\theta}_{i,t}$, where the last inequality follows from \eqref{eqn:Byz_src_uppbnd2}. This establishes our claim regarding \eqref{eqn:Byz_bound1}. Now consider the update of $\mu_{i,t+1}(\theta)$ based on \eqref{eqn:rule2}, when $t\geq \tilde{t}_2+1$.  In light of the above arguments, the numerator of the fraction on the RHS of \eqref{eqn:rule2} is upper-bounded by $C_1(\omega)e^{-(\tilde{K}(\theta^{\star},\theta)-\epsilon)t}$, while the denominator is lower-bounded by $\eta(\omega)$. We conclude that for all $i \in \mathcal{L}_1{(\theta^{\star},\theta)}\cap\mathcal{R}$:
\begin{equation}
\mu_{i,t}(\theta) < {(C_1(\omega))}^2C_2(\theta,\epsilon)e^{-(\tilde{K}(\theta^{\star},\theta)-\epsilon)t}, \forall t \geq \tilde{t}_2+2,
\label{eqn:Byz_level1bnd}
\end{equation}
where $C_2(\theta,\epsilon)=e^{(\tilde{K}(\theta^{\star},\theta)-\epsilon)}$. 
With $\mathcal{L}_0{(\theta^{\star},\theta)}\triangleq\mathcal{S}(\theta^{\star},\theta)$, we recursively define the sets $\mathcal{L}_r{(\theta^{\star},\theta)}, 1\leq r \leq (n-1)$ as:
\begin{equation}
\mathcal{L}_r{(\theta^{\star},\theta)}\triangleq \{i\in\mathcal{V}\setminus\bigcup_{q=0}^{r-1}\mathcal{L}_q{(\theta^{\star},\theta)} \hspace{1mm} {:} \hspace{1mm}  |\mathcal{N}_i\cap \{\bigcup_{q=0}^{r-1}\mathcal{L}_q{(\theta^{\star},\theta)}\}| \geq (2f+1)\}.
\end{equation}
We claim that the following is true for all $i\in\mathcal{L}_r{(\theta^{\star},\theta)}\cap\mathcal{R}$: 
\begin{equation}
\mu_{i,t}(\theta) < {(C_1(\omega))}^{r+1}{(C_2(\theta,\epsilon))}^re^{-(\tilde{K}(\theta^{\star},\theta)-\epsilon)t}, \forall t \geq \tilde{t}_2+(r+1).
\label{eqn:Byz_levelrbnd}
\end{equation}
To prove the claim, we proceed via induction on $r$. The base cases when $r\in\{0,1\}$ have already been established. Suppose equation \eqref{eqn:Byz_levelrbnd} holds for all $r\in\{0,\ldots,m-1\}$, where $m\in\{2,\ldots,n-1\}.$ 
The claim easily extends to the case when $r=m$ by noting that (i)  $\mathcal{L}_m{(\theta^{\star},\theta)}$ is non-empty if $\mathcal{V}\setminus\{\bigcup_{q=0}^{(m-1)}\mathcal{L}_q{(\theta^{\star},\theta)}\}$ is non-empty (based on the hypothesis of the theorem), (ii) any agent $i\in\mathcal{L}_m{(\theta^{\star},\theta)}\cap\mathcal{R}$ has at least $(2f+1)$ neighbors in the set $\bigcup_{q=0}^{(m-1)}\mathcal{L}_q{(\theta^{\star},\theta)}$, of which at least $f+1$ are regular (based on the $f$-locality of the adversarial model), and (iii) using the induction hypothesis and arguments similar to those used to arrive at \eqref{eqn:Byz_level1bnd}. We have thus verified the correctness of \eqref{eqn:Byz_levelrbnd}. Now taking the natural log on both sides of \eqref{eqn:Byz_levelrbnd}, dividing throughout by $t$, simplifying, and then taking the limit inferior on both sides of the resulting inequality immediately leads to \eqref{eqn:rate_Byz}. Finally, to complete the proof, it suffices to note that $\bigcup_{q=0}^{(n-1)}\mathcal{L}_q{(\theta^{\star},\theta)}=\mathcal{R}$.
\end{proof}
\bibliographystyle{IEEEtran} 
\bibliography{refs}
\end{document}